\documentclass[conference]{IEEEtran}
\usepackage{graphicx}
\usepackage{amsthm}
\usepackage{epsfig}
\usepackage{latexsym}
\usepackage{amsfonts}
\usepackage{here}
\usepackage{rawfonts}
\usepackage[latin1]{inputenc}
\usepackage[T1]{fontenc}
\usepackage{calc}
\usepackage{capitalgreekitalic}
\usepackage{url}
\usepackage{enumerate}
\usepackage{color}
\usepackage[tbtags]{amsmath}
\usepackage{amssymb}
\usepackage{upref}
\usepackage{epic,eepic}
\usepackage{times}
\usepackage{dsfont}
\usepackage{comment}
\usepackage{cite}
\usepackage{bbm} 

\renewcommand{\Pr}{\ensuremath{\operatorname{Pr}}}

\newtheorem{theorem}{\bf Theorem}

\usepackage{dsfont}

\newcounter{step}
\newlength{\totlinewidth}
  {\end{list}%
  \rule{\linewidth}{1pt}}
\newcounter{substep}

  {\end{list}}

\newlength{\aligntop}
\newlength{\alignbot}
 \makeatletter
\renewenvironment{align}{%
  \vspace{\aligntop}
  \start@align\@ne\st@rredfalse\m@ne
}{%
  \math@cr \black@\totwidth@
  \egroup
  \ifingather@
    \restorealignstate@
    \egroup
    \nonumber
    \ifnum0=`{\fi\iffalse}\fi
  \else
    $$%
  \fi
  \ignorespacesafterend%
  \vspace{\alignbot}\par\noindent
} \makeatother

\IEEEoverridecommandlockouts

\usepackage{algorithm}
\usepackage{algorithmic}

\makeatletter
\newcommand\semihuge{\@setfontsize\semihuge{19.3}{25}}
\makeatother

\makeatletter
\newcommand\semismall{\@setfontsize\semihuge{12.4}{15}}
\makeatother

\usepackage{subfigure}

\begin{document}

\title{\semihuge Echo-Liquid State Deep Learning for $360^\circ$ Content Transmission and Caching in Wireless VR Networks with Cellular-Connected UAVs}

\author{\Large{Mingzhe Chen\IEEEauthorrefmark{1}}, Walid Saad\IEEEauthorrefmark{2}, and Changchuan Yin\IEEEauthorrefmark{1}\vspace*{0.3em}\\ 
\authorblockA{\small \IEEEauthorrefmark{1}Beijing Laboratory of Advanced Information Network, Beijing University of Posts and Telecommunications, Beijing, China 100876,\\ Emails: \protect\url{chenmingzhe@bupt.edu.cn}, \protect\url{ccyin@ieee.org.} \\
\IEEEauthorrefmark{2}Wireless@VT, Bradley Department of Electrical and Computer Engineering, Virginia Tech, Blacksburg, VA, USA, Email: \protect\url{walids@vt.edu.}\\
}\vspace*{-2em}
 }

\maketitle
%
\begin{abstract}
In this paper, the problem of content caching and transmission is studied for a wireless virtual reality (VR) network in which unmanned aerial vehicles (UAVs) capture videos on live games or sceneries and transmit them to small base stations (SBSs) that service the VR users. 
However, due to its limited capacity, the wireless network may not be able to meet the delay requirements of such $360^\circ$ content transmissions. To meet the VR delay requirements, the UAVs can extract specific visible content (e.g., user field of view) from the original $360^\circ$ data and send this visible content to the users so as to reduce the traffic load over backhaul and radio access links. To further alleviate the UAV-SBS backhaul traffic, the SBSs can also cache the popular contents that users request. This joint content caching and transmission problem is formulated as an optimization problem whose goal is to maximize the users' reliability, defined as the probability that the content transmission delay of each user satisfies the instantaneous VR delay target.
 To address this problem, a distributed deep learning algorithm that brings together new neural network ideas from liquid state machine (LSM) and echo state networks (ESNs) is proposed. The proposed algorithm enables each SBS to predict the users' reliability so as to find the optimal contents to cache and content transmission format for each UAV.
  Analytical results are derived to expose the various network factors that impact content caching and content transmission format selection.
Simulation results show that the proposed algorithm yields $25.4\%$ gain of reliability compared to Q-learning. The results also show that the proposed algorithm can achieve $14.7\%$ gain of reliability due to the reduction of traffic load over backhaul compared to the proposed algorithm with random caching.
\end{abstract}

%
%
\section{Introduction}
Deploying virtual reality (VR) systems over wireless cellular networks will enable VR users to experience and interact with virtual and immersive environments without geographical or behavioral restrictions \cite{7946930}. However, for each $360^\circ$ VR content transmission, the data rate requirement may exceed 50 Mbps per user \cite{7946930} and the per-user delay requirement will be less than 20 ms. In consequence, existing cellular networks with limited backhaul and fronthaul capacity cannot readily support this $360^\circ$ VR content transmission especially for scenarios with dense VR users. {In essence, for a given user, a $360^\circ$ VR content can be divided into \emph{visible} and \emph{invisible} components. The component of a $360^\circ$ content that is visible to a given user (e.g., in the user's field of view) is defined as a \emph{visible content}.} 
One promising approach to meet the low latency and high data rate needs of VR applications is to transmit each user's visible content extracted from the $360^\circ$ content so as to reduce the backhaul and fronthaul load. 
To further alleviate the backhaul traffic, the network can cache popular VR content at different base stations \cite{7973055}. However, visible content transmission and content caching in VR networks also face many challenges such as effective visible content extraction, choice of the $360^\circ$ and visible contents to cache as well as the prediction of content popularity.


The existing literature has studied a number of problems related to caching and visible contents extraction and transmission in VR systems and cellular networks such as in \cite{7946930,schaufler1996three,chakareski2017vr,8108779,8254370,7875131,7859260,xu2017overcoming,8107708, 7973055,8299576,mozaffari2018tutorial,7438747}. In \cite{7946930}, the authors introduced the challenges and benefits of caching in wireless VR networks. However, this work relies on simple toy examples and does not provide a rigorous analytical treatment of caching over wireless VR networks. 
The authors in \cite{schaufler1996three} proposed content caching methods for three dimensional VR images.  In \cite{chakareski2017vr}, the authors proposed a new approach for cached content replacement that allows optimization of the transmission delay. However, the works in \cite{schaufler1996three} and \cite{chakareski2017vr} focus only on caching without considering the time that the VR controller needs to process the VR images which may significantly reduce the quality-of-experience (QoE) of VR users. The works in \cite{8108779,8254370,7875131,7859260,xu2017overcoming,8107708,7973055,8299576,mozaffari2018tutorial,7438747} studied a number of problems related to the application of caching for wireless networks such as coded caching, cached content replacement, proactive caching, cache-enabled unmanned aerial vehicles (UAVs), and caching for millimeter wave communications. However, these existing works \cite{8108779,8254370,7875131,7859260,xu2017overcoming,8107708,7973055,8299576,mozaffari2018tutorial,7438747} focus on caching of the contents that have equal data size and do not consider joint caching of visible and $360^\circ$ contents to offload the backhaul traffic. 
Moreover, the works in \cite{8108779,8254370,7875131,7859260,xu2017overcoming,8107708,7973055,8299576,mozaffari2018tutorial,7438747} do not consider the limited computational resources that each small base station (SBS) can use for caching. In fact, in order to cache VR content, the SBSs need to first process it which can incur computational delays. In particular, for VR, if the SBSs want to store the visible content, they need to use enough computational resources to extract the visible contents from the original $360^\circ$ contents. 

The main contribution of this paper is a novel framework for enabling caching and visible content transmission for wireless VR networks so as to reduce the traffic over the backhaul and enable VR users achieve their delay requirements. To our best knowledge, \emph{this is the first work that jointly considers visible and $360^\circ$ content caching, visible content transmission, and processing time used for visible content extraction.} Hence, our key contributions include:
\begin{itemize}
\item We propose a novel VR model in which UAVs are used to collect the contents that users request and send the collected contents to cache-enabled SBSs that can serve the VR users. 


\item For the considered VR applications, we analyze how the SBSs can {make decisions related to caching and content transmission over the UAV-SBS
backhaul \cite{Dongtobepublished}, while determining the format
of both cached and transmitted content ($360^\circ$ or visible content).}
  We formulate this joint transmission and caching problem as an optimization problem whose goal is to maximize the reliability of VR users.
  
\item To solve this problem, we propose a deep learning algorithm that combines liquid state machine (LSM) \cite{maass2010liquid} spiking neural networks with echo state networks (ESNs) \cite{Minimum} to determine the transmission and caching strategies by predicting the VR users' reliability.
 In contrast to existing algorithms that use LSM or ESN (e.g., \cite{7875131} and \cite{Chen2016Echo}), the proposed framework: a) can use historical information to find the relationship between the users' reliability, caching, and content transmission format, b) has a larger capacity memory to record network information such as users' content requests, and c) has lower complexity  for training. In fact, since the proposed algorithm uses a deep architecture with both LSM and ESN, existing results on shallow neural networks \cite{7875131} and \cite{Chen2016Echo}, cannot be directly applied.

\item We perform fundamental analysis to find the optimal content transmission format and cached contents. The analytical results show that the cached contents of each SBS depend on the data rates of backhaul and the users, the number of users, the computational resources allocated to each user, and the cache storage size. Meanwhile, the content transmission format depends on the computational resources of each UAV and SBS, the data size of each content, and the number of users that request the same content. 
\item Simulation results show that the proposed algorithm can yield {25.4\%} gain in terms of uses' reliability compared to Q-learning. 
\end{itemize}

The rest of this paper is organized as follows. The system model and problem formulation are described in Section \uppercase\expandafter{\romannumeral2}. The ESN and LSM based algorithm for reliability maximization is proposed in Section \uppercase\expandafter{\romannumeral3}. Simulation results are presented and analyzed in Section \uppercase\expandafter{\romannumeral4}. Conclusions are drawn in Section \uppercase\expandafter{\romannumeral5}.

%
%
%
%

\section{System Model and Problem Formulation}\label{se:system}

%
%
\begin{figure}[!t]
  \begin{center}
   \vspace{0cm}
    \includegraphics[width=9cm]{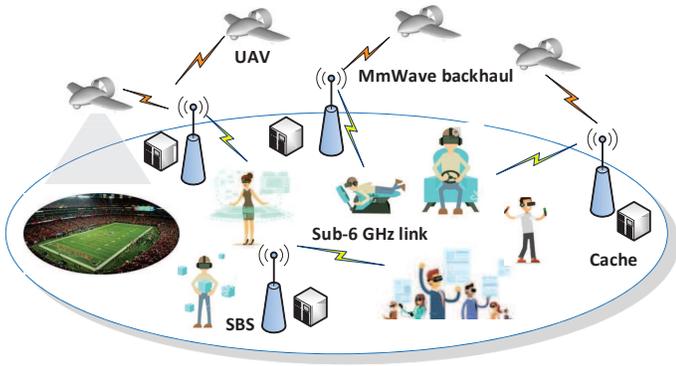}
    \vspace{-0.2cm}
    \caption{\label{model1} The architecture of a VR network that consists of cache-enabled SBSs, VR users, and UAVs.}
  \end{center}\vspace{-0.7cm}
\end{figure}

\begin{figure}[!t]
  \begin{center}
   \vspace{0cm}
    \includegraphics[width=9cm]{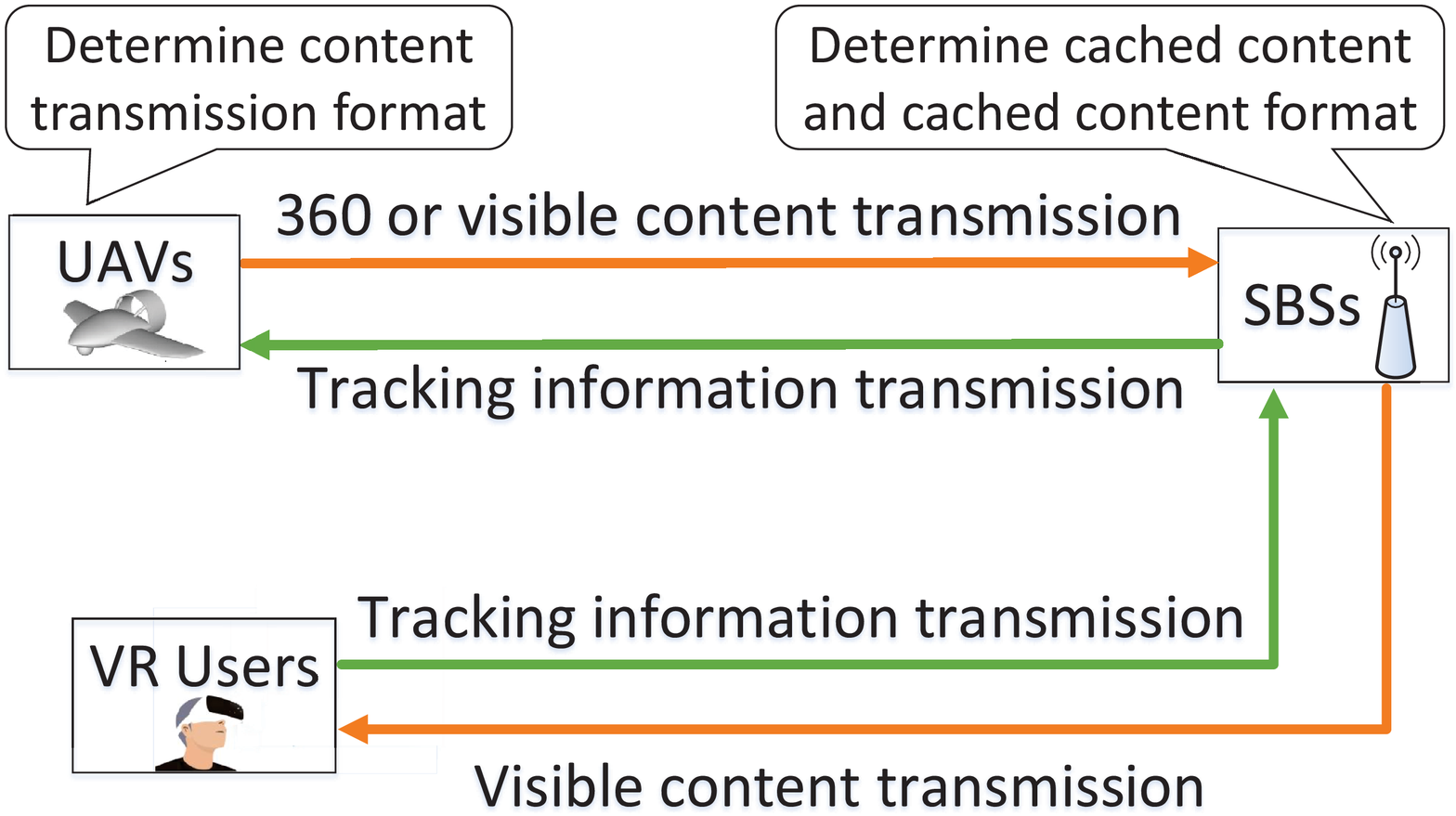}
    \vspace{-0.2cm}
    \caption{\label{model} The content and tracking information transmissions in a VR network.}
  \end{center}\vspace{-0.7cm}
\end{figure}

Consider a cellular network that encompasses a set $\mathcal{V}$ of $V$ UAVs and a set $\mathcal{B}$ of $B$ SBSs serving a set $\mathcal{U}$ of $U$ wireless VR users, as shown in Fig. \ref{model1}. In this model, the UAVs are equipped with cameras and used to collect the VR contents that will be transmitted to the SBSs via \emph{wireless backhaul links}. Each UAV $k$ can provide a set $\mathcal{C}_k$ of $C_k$ contents for VR users. VR tracking information that consists of each users' position and orientation information is needed for visible content extraction. This tracking data is transmitted in the uplink from the users to the SBSs, then from the SBSs to the UAVs via cellular links. Here, the \emph{visible content} consists of $120^\circ$ horizontal and $120^\circ$ vertical images \cite{chakareski2017vr}, while an original VR content with $360^\circ$ images is referred to as a $360^\circ$ content. 
 VR users that request the same content will have the same $360^\circ$ contents but different visible contents.
 For each content that they transmit to the SBSs, the UAVs can select the appropriate content format (visible or $360^\circ$ content).    


     
\subsection{Caching Model}
Caching at the SBSs, referred as SBS cache hereinafter, allows the SBSs to directly transmit their cached contents to VR users without any backhaul transmission. 
Each SBS cache can store up to $S$ Mbits of data. We let $\mathcal{S}_j$ be the set of contents that are stored at SBS $j$. Each content $a$ requested by a user associated with SBS $j$ can be stored in two formats: visible format or $360^\circ$ format. We let $f_{ja}$ be the format of cached content $a$. $f_{ja} \in \left\{{120^\circ}, {360^\circ}\right\}$. $f_{ja}=360^\circ$ indicates that SBS $j$ stores a $360^\circ$ content $a$, otherwise, $f_{ja}=120^\circ$ indicates that SBS $j$ stores a visible content $a$. The caching constraint of SBS $j$ is:
\begin{equation}\label{eq:chi}
\sum\limits_{a \in {\mathcal{S}_j},} { \chi_a \left( {{f_{ja}}} \right)}  \le S, {\chi _a}\left( {{f_{ja}}} \right) = \left\{ {\begin{array}{*{20}{c}}
{{G_{360^\circ}}, {f_{ja}} = {{360}^ \circ },}\\
{{U_jG_{120^\circ},}\;\; {f_{ja}} = {{120}^ \circ },}
\end{array}} \right.
\end{equation}
%
where $G_{120^\circ}$ and $G_{360^\circ}$ represent the data sizes of a visible content and a $360^\circ$ content, respectively. $U_j$ is the number of users associated with SBS $j$. From (\ref{eq:chi}), we can see that, a given SBS will have to store only one $360^\circ$ content $a$ for all of its associated users. In contrast, if SBS $j$ stores visible contents, then it must store a specific visible content for each user associated with it (e.g., corresponding to the field of view of each user). This is due to the fact that if an SBS stores a $360^\circ$ content $a$, it can extract all possible visible contents from the stored $360^\circ$ contents.   

\subsection{Transmission Model}
As shown in Fig. \ref{model}, contents can be sent from: a) UAVs to the SBSs, and then from the SBSs to the users, and b) SBS cache to the users. Next, we introduce the models for backhaul (UAVs-to-SBSs links) and SBSs-to-users links.  

\subsubsection{Backhaul Links}
Downlink backhaul links are used for VR content transmission from the UAVs to SBSs. Since the UAVs need to transmit large-sized VR content and have high probability of building line-of-sight (LoS) links, millimeter wave (mmWave) with large bandwidth will be a good choice for the UAV-SBS backhaul links. The standard log-normal shadowing model of \cite{WirelessCommunications} is used to model the mmWave propagation channel of the  downlink backhaul links. The standard log-normal shadowing model can be used to model the LoS and non-line-of-sight (NLoS) links by choosing specific channel parameters. Thus, the LoS and NLoS path loss experienced by UAV $k$ transmitting a content to SBS $j$ is \cite{BroadbandMillimeterWave2013} (in dB):
\begin{equation}
\begin{split}
&\!\!l_{kj}^{\textrm{LoS}} ={L_{FS}}\left( {{d_0}} \right) 
+ 10\mu_{\textrm{LoS}}\log \left( {{d_{kj}}} \right) + {\chi _{\sigma_\textrm{LoS}} },
\end{split}
\end{equation}
\begin{equation}
\begin{split}
&\!\!l_{kj}^{\textrm{NLoS}} ={L_{FS}}\left( {{d_0}} \right) \!+ \!10\mu_{\textrm{NLoS}}\log \left( {{d_{kj}}} \right)\! +\! {\chi _{\sigma_\textrm{NLoS}} },
\end{split}
\end{equation}
where $L_{FS}\left(d_0\right)$ is the free space path loss given by $20\log \left( {{{d_0f_c4\pi } \mathord{\left/
 {\vphantom {{4\pi } c}} \right.
 \kern-\nulldelimiterspace} c}} \right)$ with $d_0$ being the free-space reference distance, $f_c$ being the carrier frequency, and $c$ being the speed of light. ${{d_{kj}}}$ is the distance between UAV $k$ and SBS $j$. $\mu_\textrm{LoS}$ and $\mu_\textrm{NLoS}$ are the path loss exponents for LoS and NLoS links. ${\chi _{\sigma_\textrm{LoS}} }$ and ${\chi _{\sigma_\textrm{NLoS}} }$ are the shadowing random variables which are, respectively, represented as the Gaussian random variables in dB with zero mean and $\sigma_\textrm{LoS}$, $\sigma_\textrm{NLoS}$ dB standard deviations. 
The LoS probability is given by \cite{Modelingairtoground}:
\begin{equation}\label{eq:propability}
\Pr \left( {l_{t,ki}^\textrm{LoS}} \right) = {\left( {1 + X\exp \left( { - Y\left[ {\phi_t   - X} \right]} \right)} \right)^{ - 1}},
\end{equation}
where $X$ and $Y$ are constants which depend on the environment (e.g., rural or urban) and $\phi_t  = {\sin ^{ - 1}}\left( {{h_{k}}/d_{kj}} \right)$ is the elevation angle with $h_k$ being the altitude of UAV $k$. The average path loss from UAV $k$ to SBS $j$ is \cite{Modelingairtoground}:
\begin{equation}\label{eq:averl}
\bar l_{kj} = \Pr \left( {l_{kj}^\text{LoS}} \right) \times {l_{kj}^\text{LoS}}+ \Pr \left( {{l_{kj}^\textrm{NLoS}}} \right) \times {l_{kj}^\textrm{NLoS}},
\end{equation}
where $\Pr \left( {l_{kj}^\textrm{NLoS}}\right)=1-\Pr \left( {l_{kj}^\textrm{LoS}}\right)$.  
The average signal-to-noise ratio (SNR)
of SBS $j$ is given by:
 \begin{equation}\label{eq:sir}
{\gamma _{kj}^{\textrm{VD}}} =\frac{{{P_{V}}}}{{{{10}^{{{{\bar l_{kj}}} \mathord{\left/
 {\vphantom {{{l_{t,ki}}\left( \boldsymbol{V}_{\tau,t,ki} \right)} {10}}} \right.
 \kern-\nulldelimiterspace} {10}}}}{\sigma ^2}}},
\end{equation}
where $P_{V}$ is the transmit power of UAV $k$ which is assumed to be equal for all UAVs and $\sigma^2$ is the variance of the Gaussian noise. We assume that the bandwidth of the transmission link from UAV $k$ to SBS $j$ is $B^\textrm{VD}$. For each user associated with SBS $j$, the average channel capacity between UAV $k$ and SBS $j$ will be {${c_{kj}^\textrm{VD}}= {\frac{{{B^\textrm{VD}}}}{{{U_j}}}{{\log }_2}\left( {1 + {\gamma _{kj}^{\textrm{VD}}}} \right)}$}, where $U_j$ is the number of users associated with SBS $j$. $\frac{{{B^\textrm{VD}}}}{{{U_j}}}$ implies that the bandwidth used for the content transmission between a UAV and SBS $j$ is equally allocated to the transmissions of users associated with SBS $j$. 

Uplink backhaul links are used to transmit the tracking information of the users from SBSs to the UAVs. Since the sub-6 GHz band can provide a more reliable transmission and a smaller path loss compared to the mmWave channel \cite{semiari2018integrated}, we consider probabilistic LoS and NLoS links over the {sub-6 GHz} band for tracking information transmission.  
The LoS and NLoS path loss from SBS $j$ to UAV $k$ is \cite{Mozaffari2016Unmanned}:
\begin{equation}
{L_{kj}^\text{LoS}} ={{d_{kj}^{-\beta}}}, {L_{kj}^\text{NLoS}}=\eta{{d_{kj}^{-\beta}}},
\end{equation}
where $\beta$ is the path loss exponent. In our model, each SBS $j$ uses a dedicated channel to transmit the tracking information and, hence, we ignore the interference between SBSs, over the backhaul. The LoS connection probability, the average path loss, and the average SNR of the link from SBS $j$ to UAV $k$ can be calculated using (\ref{eq:propability})-(\ref{eq:sir}).
For each SBS $j$, the bandwidth used for tracking information transmission is $B^\textrm{VU}$. Hence, for each user associated with SBS $j$, the average channel capacity of the SBS-UAV link is ${c_{kj}^\textrm{VU}}= {\frac{{{B^\textrm{VU}}}}{{{U_j}}}{{\log }_2}\left( {1 + {\gamma _{kj}^{\textrm{VU}}}} \right)}$ where $\gamma _{kj}^{\textrm{VU}}$ is the average SNR of the link from SBS $j$ to UAV $k$.

\subsubsection{SBSs-Users Links}       
For each SBS $j$, the total bandwidth used for content transmission from each SBS to its associated users is $B^\textrm{SD}$. The channel capacity of the link from SBS $j$ to user $i$ can be given by: 
 \begin{equation}
 c_{ji}^\textrm{SD}\!=\! \frac{{{B^\textrm{SD}}}}{{{U_j}}}{\log _2}\left(\!1\!+\! {\frac{{{P_B}{ h_{ij}}}}{{\sum\limits_{n \in \mathcal{B}, n \ne j} {{P_B}{ h_{in}}}  \!+\! {\sigma ^2}}}} \!\right)\!\!,
\end{equation}   
where $P_B$ is the transmit power of each SBS $j$ and $ h_{ij}$ is the channel gain between SBS $j$ and user $i$. 
Similarly, the channel capacity of the link from user $i$ to SBS $j$ is:
 \begin{equation}
 c_{ji}^\textrm{SU}\!=\! \frac{{{B^\textrm{SU}}}}{{{U_j}}}{\log _2}\left(\!1\!+\! {\frac{{{P_U}{ h_{ij}}}}{{\sum\limits_{n \in \mathcal{U}, n \ne i} {{P_U}{ h_{nj}}}  \!+\! {\sigma ^2}}}} \!\right)\!\!,
\end{equation} 
where ${B^\textrm{SU}}$ is the total bandwidth used for tracking information transmission from the users to the SBSs.

\subsubsection{Transmission delay} 
 In our model, the UAVs can select the appropriate transmission format (visible or $360^\circ$ format) for each content. The transmission format of each content that user $i$ requests is defined as $\boldsymbol{g}_i=\left[g_{i1}, \ldots, g_{iN}\right]$ where $g_{ia} \in \left\{ 120^\circ , 360^\circ\right\}$. $g_{ia}=360^\circ$ indicates that a UAV transmits $360^\circ$ content $a$ to the SBS that associates with user $i$ while $g_{ia}=120^\circ$ implies that a UAV transmit visible content $a$ to that SBS. For user $i$ associated with SBS $j$, the time used for the transmissions of a content $a$ and tracking information can be given by:
 \begin{equation}\label{eq:delay}
 \begin{split}
T_i^\textrm{M}\!\!\left( {a,{\boldsymbol{g}_i},{\mathcal{S}_j},{\mathcal{U}_j}}\! \right) =&\\&\!\!\!\!\!\!\!\!\!\!\!\!\!\!\!\!\!\!\!\!\!\!\!\!\!\!\!\!\!\!\!\!\!\!\!\!\!\!\!\!\!\!\!\!\!
\left\{ {\begin{array}{*{20}{c}}
  {\frac{{G_{{120}^ \circ}}}{{c_{kj}^\textrm{VD}}} + \frac{A}{{c_{kj}^\textrm{VU}}} + \frac{{G_{{120}^ \circ}}}{{c_{ji}^\textrm{SD}}} + \frac{A}{{c_{ji}^\textrm{SU}}},a \notin {\mathcal{S}_j},{g_{ia}} = {{120}^ \circ },} \\ 
  {\frac{{G_{{360}^ \circ}{}}}{{U_{ja}c_{kj}^\textrm{VD}}} + \frac{A}{{c_{kj}^\textrm{VU}}} + \frac{{G_{{120}^ \circ}}}{{c_{ji}^\textrm{SD}}} + \frac{A}{{c_{ji}^\textrm{SU}}},a \notin {\mathcal{S}_j},{g_{ia}} = {{360}^ \circ },} \\ 
  {\frac{{G_{{120}^ \circ}}}{{c_{ji}^\textrm{SD}}} + \frac{A}{{c_{ji}^\textrm{SU}}}, \;\;\;\;\; \;\;\;\;\;\;\;\;\;a \in {\mathcal{S}_j},\;\;\;\;\;\;\;} 
\end{array}} \right.
 \end{split}
%
 \end{equation}
 where $\mathcal{U}_j$ is the set of users associated with SBS $j$, $A$ is the data size of tracking information, $U_{ja}$ is the number of users that request content $a$, ${\frac{{G_{{120}^ \circ}}}{{c_{kj}^\textrm{VD}}}}$ and $ \frac{{G_{{360}^ \circ}}}{{U_{ja}c_{kj}^\textrm{VD}}} $ represent the time used for content transmission from UAVs to SBS $j$, and ${\frac{{{G_{{120}^ \circ}}}}{{c_{ji}^\textrm{SD}}}}$ represents the time used for content transmission from SBS $j$ to user $i$. ${\frac{{A}}{{c_{kj}^\textrm{VU}}}}$ and ${\frac{{A}}{{c_{ji}^\textrm{SU}}}}$ represent the time used for tracking information transmission from SBS $j$ to UAV $k$ and from user $i$ to SBS $j$.  
 From (\ref{eq:delay}), we can see that, the transmission delay depends on the user association $\mathcal{U}_j$, cached content set $\mathcal{S}_j$, and content transmission format vector $\boldsymbol{g}_i$.  
 
%
%

\subsection{Computing Model}
 When a UAV wants to transmit a visible content to the SBSs, it must use the tracking information collected from the users to generate the visible content. For user $i$ associated with SBS $j$, the size of the data that needs to be processed when a UAV extracts visible content from a $360^\circ$ content $a$ is $H_{ia}$. The processing time that a UAV uses to generate visible content $a$ is:
 \begin{equation}\label{eq:TiU}
 T_{i}^\textrm{UP}\left(a,\boldsymbol{g}_{i}, U_j,\mathcal{S}_j\right)={ \frac{{{H_{ia}}}}{{{R_{\textrm{U}} \mathord{\left/
 {\vphantom {R {{U_j}}}} \right.
 \kern-\nulldelimiterspace} {{U_j}}}}} \mathbbm{1}_{\left\{ {g_{ia}} = {{120}^ \circ }, a \notin \mathcal{S}_j\right\}}},
 \end{equation}
 where $\mathds{1}_{\left\{x\right\}}=1$ when $x$ is true and $\mathds{1}_{\left\{x\right\}}=0$ otherwise. $R_\textrm{U}$ is the total number of computational resources that a UAV uses to extract visible contents for each SBS.
 From (\ref{eq:TiU}), we can see that if a UAV wants to transmit a visible content $a$ to the users, it must spend time to extract visible content from $360^\circ$ content $a$.
 
To reduce the traffic load over SBS-users links, the SBSs will only transmit visible contents to the users. In consequence, when an SBS stores $360^\circ$ content $a$, it must use tracking information to generate the visible content then transmit the visible content to the users. The processing time that SBS $j$ uses to generate visible content $a$ of user $i$ is given by:
  \begin{equation}\label{eq:TISP}
 T_{i}^\textrm{SP}\!\left(a,\boldsymbol{g}_{i},\boldsymbol{f}_j, U_j, \mathcal{S}_j\right)\!=\!\frac{{{H_{ia}}}}{{{R_\textrm{S} \mathord{\left/
 {\vphantom {R {{U_j}}}} \right.
 \kern-\nulldelimiterspace} {{U_j}}}}} \mathbbm{1}_{\left\{ a \in \mathcal{S}_j, f_{ja}=360^\circ~\textrm{or}~a \notin \mathcal{S}_j, g_{ia}=360^\circ  \right\}},
 \end{equation} 
 where $R_\textrm{S}$ is the total number of computational resources allocated to each SBS. (\ref{eq:TISP}) shows that SBS $j$ needs to spend time for the extraction of visible content $a$ in two cases: a) $360^\circ$ content $a$ is stored in the cache and b) content $a$ is not stored in the cache and the UAVs transmit $360^\circ$ content $a$ to SBS $j$. The total processing time of each user $i$ is $T_{i}^\textrm{P}\left(a,\boldsymbol{g}_{i}, \boldsymbol{f}_j, U_j, \mathcal{S}_j\right)= T_{i}^\textrm{SP}\left(a, \boldsymbol{g}_{i}, \boldsymbol{f}_j, U_j, \mathcal{S}_j\right)+T_{i}^\textrm{UP}\left(a, \boldsymbol{g}_{i}, U_j, \mathcal{S}_j\right).$

\subsection{Problem Formulation}
Our goal is to develop an effective transmission strategy for UAVs as well as an effective caching strategy and user association scheme for SBSs so as to maximize the reliability of each VR user. The reliability of each user $i$ is defined as follows:
\begin{equation}\label{eq:Pi}
\begin{split}
&{\mathbbm{P}_i}\left( \boldsymbol{g}_{i}, \boldsymbol{f}_j, \mathcal{U}_j, \mathcal{S}_j \right) =\\& \mathop {\lim }\limits_{T \to \infty } \frac{1}{T}\sum\limits_{t = 1}^T {{\mathbbm{1}_{\left\{ {T_{it}^\textrm{P}\left( a_{it}, \boldsymbol{g}_{i}, \boldsymbol{f}_j, U_j, \mathcal{S}_j \right) + T_{it}^\textrm{M}\left( {a_{it},{\boldsymbol{g}_i},{\mathcal{S}_j},{\mathcal{U}_j}} \right) \le D} \right\}}}},
\end{split}
\end{equation}
where $a_{it}$ is the content that user $i$ requests at time $t$. $T_{it}^\textrm{P}$ and $T_{it}^\textrm{M}$ are, respectively, the processing time and transmission time of user $i$ at time $t$. $D$ is the delay requirement of each VR user. (\ref{eq:Pi}) actually captures the probability of user $i$'s successful transmission which means that the total processing and transmission delays of user $i$ meet the delay requirement. 
Based on (\ref{eq:Pi}), we formulate an optimization problem whose objective is to maximize the reliability of all users. This maximization problem involves determining the transmission format $\boldsymbol{g}_i$, user association policy ${\mathcal{U}_j}$, and set of cached contents ${\mathcal{S}_j}$ as well as cached content format $\boldsymbol{f}_{j}$. Therefore, this problem can be formalized as follows:
\addtocounter{equation}{0}
\begin{equation}\label{eq:max}
\begin{split}
\mathop {\max }\limits_{\boldsymbol{g}_i,\boldsymbol{f}_j,{\mathcal{U}_j},\mathcal{S}_j}   \sum\limits_{j \in \mathcal{B}}\sum\limits_{i \in \mathcal{U}_j}  &{\mathbbm{P}_i}\!\left( \boldsymbol{g}_{i}, \boldsymbol{f}_j, \mathcal{U}_j, \mathcal{S}_j \right)
\end{split}
\end{equation}
\vspace{-0.3cm}
\begin{align}\label{c1}
\setlength{\abovedisplayskip}{-20 pt}
\setlength{\belowdisplayskip}{-20 pt}
&\;\;\;\;\rm{s.\;t.}\scalebox{1}{$ \sum\limits_{a \in {\mathcal{S}_j},} { \chi_a \left( {{f_{ja}}} \right)}  \le S,\;\;\;\;\;\forall j \in \mathcal{B},$}\tag{\theequation a}\\
&\scalebox{1}{$\;\;\;\;\;\;\;\;\;\; f_{ja} \in \left\{{120^\circ}, {360^\circ} \right\},\;\;\;\forall j \in \mathcal{B}, a \in \mathcal{S}_j, $} \tag{\theequation b}\\
&\scalebox{1}{$\;\;\;\;\;\;\;\; \;\; g_{ia} \in \left\{ {120^\circ}, {360^\circ} \right\},\;\;\;\forall i \in \mathcal{U}_j, a \in \mathcal{C}_k, a \notin \mathcal{S}_j,$}   \tag{\theequation c}
\end{align}
where {(14a) is the maximum cache storage size.} (14b) indicates that each SBS $j$ can store visible or $360^\circ$ content $a$. (14c) indicates that each UAV can transmit visible or $360^\circ$ content $a$ to user $i$. From (\ref{eq:max}), we can see that user association ${\mathcal{U}_j}$, cached content set ${\mathcal{S}_j}$, cached content format $\boldsymbol{f}_j$, and content transmission format $\boldsymbol{g}_j$ are coupled. Meanwhile, all of the elements in vectors $\boldsymbol{f}_j$ and $\boldsymbol{g}_j$ are discrete. In consequence, the problem in (\ref{eq:max}) is challenging to solve by conventional optimization algorithms. Hence, we propose a deep reinforcement learning (RL) algorithm \cite{li2017deep} to solve it since RL algorithms can adaptively adjust the user association policy, cached contents, and cached content format according to the network and users states.     

 \section{Deep Echo Liquid State Machine Learning for Reliability Maximization} \label{se:al}
 Next, we introduce a novel deep RL algorithm that merges together the neural network concepts of LSM and ESN into a single \emph{deep learning framework {called echo liquid state machine (ELSM)}}, so as to perform user association and determine transmission content format, cached content set, as well as cached content format. Liquid state machines are spiking neural networks \cite{chen2017machine} that are randomly generated. When used for learning, LSMs can store information related to the network environment over time and adjust the user association policy, cached contents, and cached content formats according to the users' content requests. However, 
a traditional LSM uses feed-forward neural networks (FNNs) as the output function. Training FNNs is complex (requires computing the gradients for all of the neurons) and, hence, the complexity of training a traditional LSM is high. To address these challenges, we propose a learning approach that uses an ESN scheme as the output function of LSMs. The proposed deep learning algorithm can record large amount of historical information and use it to predict the future output. Moreover, the proposed algorithm is easy to train since ESNs only need to train the output weight matrix. We first introduce the components of the proposed ELSM algorithm. Then, we explain the use of the ELSM algorithm to solve our problem in (\ref{eq:max}).   

\begin{figure}[!t]
  \begin{center}
   \vspace{0cm}
    \includegraphics[width=9cm]{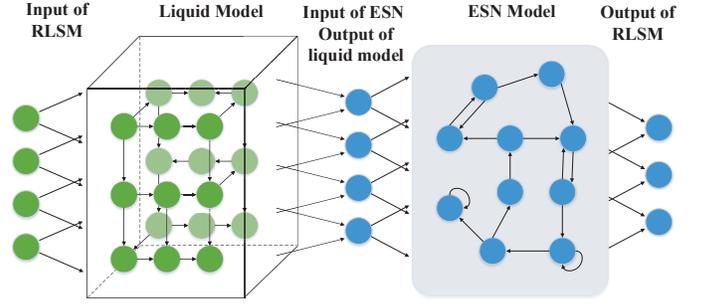}
    \vspace{-0.3cm}
    \caption{\label{algorithm} The components of the proposed deep RL algorithm using ELSM. Green and blue circles represent the spiking and normal neurons, respectively. }
  \end{center}\vspace{-0.7cm}
\end{figure}

\subsection{ESLM Deep RL Algorithm Components}
An ELSM based deep RL algorithm consists of eight components: a) agents, b) input of ELSM, c) liquid model, d) output of liquid model, e) input of ESN, f) ESN model, and g) output of ELSM, as shown in Fig. \ref{algorithm}. For each SBS $j$, the components of deep RL using ELSM are specifically defined as follows:

\begin{itemize}
\item \emph{Agents:} The agents that implement the ELSM algorithm are the SBSs. Each SBS implements one ELSM algorithm for maximizing the reliability of its associated users.
\item \emph{Input of ELSM:} The input of the ELSM deep RL algorithm is a vector $\boldsymbol{x}_{\tau,j}=\left[ {x_{\tau,j1}, \cdots , x_{\tau, jB}} \right]^{\mathrm{T}}$ where $x_{\tau,jk}$ represents the index of the policy that SBS $k$ uses for action selection at a period $\tau$ consisting of $N_\tau$ time slots. $\boldsymbol{x}_{\tau,j}$ is then used to estimate the ELSM output that captures the total reliability of the users associated with SBS $j$.  
\item \emph{Liquid Model:} The liquid model for each SBS is used to store the ELSM input information. Then, an ESN can use the historical input information to estimate the the users' reliability. The liquid model consists of $N_L=L_1\times L_2 \times L_3$ leaky integrate and fire neurons \cite{chen2017machine} that are arranged in a three dimensional-column. 

The liquid model typically has two categories of neurons: inhibitory neurons and excitatory neurons. An excitatory neuron will increase the information of the neurons in the next layer while an inhibitory neuron will decrease the information of the neurons in the next layer. The information state of each neuron $n$ at each time slot $t$ is ${v_{jn}}\left(t  \right)$, which can be calculated using \cite[Equation (13)]{maass2010liquid}.
The connections from the ELSM input to the liquid model are made with probability ${P}_\textrm{IN}$. The probability of having a connection between neurons $i$ and $j$ in the liquid model can be given by
${{P}_{ij}} = \varsigma {e^{ - {{\left( {{\zeta_{ij}  \mathord{\left/
 {\vphantom {{\zeta } \lambda }} \right.
 \kern-\nulldelimiterspace} \lambda }} \right)}^2}}}$,
where $\varsigma \in \left\{ {{\varsigma_\textrm{EE}},{\varsigma_\textrm{EI}},{\varsigma_\textrm{IE}},{\varsigma_\textrm{II}}} \right\}$ is a constant that depends on the type of both neurons. In particular, $\varsigma_\textrm{EE}$ denotes an excitatory-excitatory connection, $\varsigma_\textrm{EI}$ is an excitatory-inhibitory connection, $\varsigma_\textrm{IE}$ is an inhibitory-excitatory connection, and $\varsigma_\textrm{II}$ is a inhibitory-inhibitory connection.
$\zeta_{ij}$ is the Euclidean distance between neurons $i$ and $j$. $\lambda$ influences how often neurons are connected.

      
\item  \emph{Output of Liquid Model (Input of ESN):} The output of the liquid model at period $\tau$ is a vector  $\boldsymbol{\varphi }_j\!\left(\tau\right)\!=\! \left[ {  {\boldsymbol{v}_j}\!\left( 1\right); \ldots ; \boldsymbol{v}_j}\!\left( N_\tau  \right)  \right]$ where $\boldsymbol{\varphi}_j\left(t\right) \in   {\mathbb{R}^{N_\tau N_L \times 1}} $ and $\boldsymbol{v}_j\left(t  \right)=\left[ { {v_{j1}}\left( t \right), \ldots ,{v_{{jN_L}}}\!\left( t \right)} \right]$ with $ {v_{jn}}\left( t \right)$ being the information state of neuron $n$ at time slot $t$. In fact, the ELSM uses unique output of liquid model to represent each input information. We can also see that the size of the input vector is $B$ while the size of the output vector is $N_\tau N_L$. This implies that the liquid model can use its spiking neurons to extract and extend the dynamics of the input information. As shown in Fig. \ref{algorithm}, since the ESN is connected to the liquid model, the input of the ESN is essentially the output of liquid model.      
   

 \item  \emph{ESN Model:} The ESN model is used to find the relationship between the input and output of the ESN. The ESN model consists of an input matrix $\boldsymbol{W}_j^\textrm{in} \in \mathbb{R}^{N_W \times N_LN_\tau}$, a recurrent matrix $\boldsymbol{W}_j \in \mathbb{R}^{N_W \times N_W}$ , and an output weight matrix $\boldsymbol{W}_j^\textrm{out} \in \mathbb{R}^{N_O \times \left(N_W+N_LN_\tau\right)}$. Here, $N_W$ represents the number of neurons in ESN model and $N_O$ is the number of actions of the ELSM algorithm. 
 
 \item  \emph{Action:} Each action $i$ of the ELSM algorithm implemented by SBS $j$ is a vector $\boldsymbol{o}_{ji}\left(t\right)\!=\! \left[  {\boldsymbol{u}_{ji}}\left( t\right); \boldsymbol{q}_{ji}\left( t\right) \right]$ where ${\boldsymbol{u}_{ji}}\left( t\right)\!=\!\left[ {{u}_{ji1}}\left( t\right), \ldots,  {{u}_{jiN_{U_j}}}\left( t\right) \right]$ with ${{u}_{jik}}\left( t\right) \in \left\{0,1\right\}$. ${{u}_{jik}}\!=\!1$ implies that, under action $i$, user $k$ will be associated with SBS $j$, and ${{u}_{jik}}=0$, otherwise. $\boldsymbol{q}_{ji}\left( t\right)=\left[q_{ji1}\left( t\right), \ldots, q_{jiN} \left( t\right) \right]$ where $q_{jik} \in \left\{0, 120^\circ, 360^\circ \right\}$. $q_{jik}=0$ implies that when SBS $j$ uses action $i$, content $k$ is not stored in the cache while $q_{jik}=120^\circ$ ($q_{jik}=360^\circ$) implies that SBS $j$ stores visible ($360^\circ$) content $k$. In fact, action $\boldsymbol{o}_{ji}\left(t\right)$ determines the user association policy and cached contents at time $t$. We can also define an alternative action $\boldsymbol{o}_{ji}'\left(t\right)= \left[  {\boldsymbol{u}_{ji}}\left( t\right); \boldsymbol{z}_{ji}\left( t\right) \right]$ that determines the user association policy and content transmission format. Here, $\boldsymbol{z}_{ji}\left( t\right)=\left[z_{ji1}\left( t\right), \ldots, z_{jiN} \left( t\right) \right]$ where $z_{jik} \in \left\{ 120^\circ, 360^\circ \right\}$. $z_{jik}=120^\circ$ implies that the UAVs will transmit visible content $a$ to user $i$ while $z_{jik}=360^\circ$ implies that the UAVs will transmit $360^\circ$ content $a$ to user $i$. In fact, action $\boldsymbol{o}_{ji}\left(t\right)$ determines the user association policy and cached contents. Meanwhile, action $\boldsymbol{o}_{ji}'\left(t\right)$ decides the user association and transmission content format. The SBSs will determine the use of either $\boldsymbol{o}_{ji}\left(t\right)$ or $\boldsymbol{o}_{ji}'\left(t\right)$ according to the cache storage size, the number of users, and the number of UAV contents. Note that, for each implementation of the ELSM algorithm, each SBS will can at most use one action ($\boldsymbol{o}_{ji}\left(t\right)$ or $\boldsymbol{o}_{ji}'\left(t\right)$).    
 
  \item  \emph{Output of ELSM (output of ESN):} For each action $i$ that SBS $j$ takes, the ELSM algorithm will output the estimated total reliability for users associated with SBS $j$. Thus, the output of the ELSM algorithm is ${{y} _{j}\!\left(\boldsymbol{o}_{ji}\!\left(t\right)\!\right)}\!=\! \sum\limits_{i \in \mathcal{U}_j}  \!\!{\bar {\mathbbm{P}}_i}\!\left( \boldsymbol{o}_{ji} \left(t\right) \right)$. If the SBS use action $\boldsymbol{o}_{ji}'\left(t\right)$, then the output of the ELSM is ${{y} _{j}\left(\boldsymbol{o}_{ji}'\left(t\right)\right)}= \sum\limits_{i \in \mathcal{U}_j}  {\bar {\mathbbm{P}}_i}\!\left( \boldsymbol{o}_{ji}' \left(t\right) \right)$.
  \end{itemize}  
   To allow the ELSM algorithm to estimate the users' reliability, we need to calculate the successful transmission of users when SBS $j$ uses action $\boldsymbol{o}_{ji}\left(t\right)$ at time $t$. A transmission is said to be successful if the total processing and transmission delays of each user meet the delay requirements. The number of users that have successful transmission can be derived by the following theorem.   
  
  \begin{theorem}\label{theorem1}
\emph{Given an action $\boldsymbol{o}_{ji}\left(t\right)$ that determines user association, cached contents, and cached content format, the number of users that have successful transmissions is given by:
\begin{equation}\nonumber
\begin{split}
{{N}_\textrm{F}\!\left(\boldsymbol{o}_{ji}\!\left(t\right)\!,\!\boldsymbol{o}_{-j}\left(t\right)\right)} &\!=\!\!\!\!\!\!\!\!\!\! \sum\limits_{i \in \mathcal{U}_j, a_{it}\in \mathcal{S}_j }\!\!\!\! {{\!\!\mathbbm{1}_{\left\{\! {T_{it}^\textrm{SP}\left(a_{it}, \boldsymbol{o}_{ji}\left(t\right)\right) + \frac{{G_{{120}^ \circ}}}{{c_{ji}^\textrm{SD}}} + \frac{A}{{c_{ji}^\textrm{SU}}} \le D}\! \right\}}}}\\&\!\!\!\!\!\!\!\!\!\!\!\!\!\!\!\!\!\!\!\!\!\!\!\!\!\!\!\!\!\!+\!\!\!\!\!\!\!\!\! \sum\limits_{i \in \mathcal{U}_j, a_{it}\notin \mathcal{S}_j }\!\!\!\! \!\!\!{{\mathbbm{1}_{\left\{ {T_{it}^\textrm{P}\left( a_{it}, {g}_{ia_{it}}, \boldsymbol{o}_{ji}\left(t\right)\right) + T_{it}^\textrm{M}\left( {a_{it},{{g}_{ia_{it}}},\boldsymbol{o}_{ji}\left(t\right)} \right) \le D} \right\}}}},
\end{split}
\end{equation}}
\emph{where $\small{g_{i{a_{it}}}} \!= \!\left\{ {\begin{array}{*{20}{c}}
 \!\! \!\!{{{360}^ \circ },~\textrm{if}~G_{{360}^ \circ} \leqslant {U_{ja}}G_{{120}^ \circ},}\;\;\;\;\;\;\;\;\;\;\;\;\;\;\;\;\;\;\;\;\;\;\;\;\;\;\;\;\;\;\;\\ 
  \!\!\!\!{{{360}^ \circ },~\textrm{if}~{H_{ia}}{U_j}\!\left(\! {\frac{1}{{{R_S}}} \!- \!\frac{B}{{{R_U}}}} \right) \!+\! \left(\! {\frac{{G_{{360}^ \circ}}}{{{U_{ja}}}} \!-\! G_{{120}^ \circ}\!} \right)\! \leqslant\! 0,} \\ 
  \!\!\!\!{{{120}^ \circ },~\textrm{if}~{H_{ia}}{U_j}\!\!\left(\! {\frac{1}{{{R_S}}} \!-\! \frac{B}{{{R_U}}}} \right) \!+ \!\left( \!{\frac{{G_{{360}^ \circ}}}{{{U_{ja}}}}\! -\! G_{{120}^ \circ}\!} \right) \!> \!0,} 
\end{array}} \right.$} \emph{and $\boldsymbol{o}_{-j}\left(t\right)$ is the action vector of all SBSs except SBS $j$.} 

\end{theorem} 
  
  \begin{proof} See Appendix A.
 \end{proof}
 From Theorem \ref{theorem1}, we can see that the content transmission format depends on the data rates of backhaul and users, the computational resources of each UAV and SBS, the content size, and the number of users that request the same content. From Theorem \ref{theorem1}, we can also see that, as the number of users that request the same content increases, it will be better for the UAVs to transmit $360^\circ$ contents to the SBSs. 
 
If each SBS $j$ uses action $\boldsymbol{o}_{ji}'\left(t\right)$ for the implementation of the ELSM algorithm, the number of users that have successful transmissions at time $t$, which is given by the following theorem.
 
 \begin{theorem}\label{theorem2}
\emph{Given an action $\boldsymbol{o}_{ji}'\left(t\right)$ that determines the user association and content transmission format, the number of users that have successful transmissions is given by:}
\begin{equation}\label{eq:theorem1}
\begin{split}
&{{N}_\textrm{F}\left(\boldsymbol{o}_{ji}'\left(t\right),\boldsymbol{o}_{-j}'\left(t\right)\right)} =\\& \sum\limits_{i \in \mathcal{U}_j } {{\mathbbm{1}_{\left\{ {T_{it}^\textrm{P}\left( a_{it},\boldsymbol{f}_j, \mathcal{S}_j \boldsymbol{o}_{ji}'\left(t\right)\right) + T_{it}^\textrm{M}\left( {a_{it},\boldsymbol{f}_j, \mathcal{S}_j\boldsymbol{o}_{ji}'\left(t\right)} \right) \le D} \right\}}}},
\end{split}
\end{equation}
\emph{where $\left[ {{\mathcal{S}_j},{\boldsymbol{f}_j}} \right] = \arg \mathop {\max }\limits_{{\mathcal{S}_j},{\boldsymbol{f}_j}} \sum\limits_{{a_{it}} \in {S_j},{g_{i{a_{it}}}} = {z_{ji{a_{it}}}}} {{I_F}\left( {{f_{j{a_{it}}}}} \right)} $ with constraint $\sum\limits_{a \in {\mathcal{S}_j},} { \chi_a \left( {{f_{ja}}} \right)}  \le S$. Here, $ {I_F}\left( {{f_{j{a_{it}}}}} \right) = \left\{ {\begin{array}{*{20}{c}}
 \!\!\!\!\!\!\!\!\!\! \!\!\!\!\!\!\!\!\!\!\!\!\!\!\!\!\!\!\!\!{\sum\limits_{i \in {U_j}} {{{\mathbbm{1}_{\left\{ \frac{{G_{{120}^ \circ}}}{{c_{ji}^{SD}}} + \frac{A}{{c_{ji}^{SU}}} \leqslant D \right\}}}-N_\textrm{FU} },{f_{j{a_{it}}}} = {{120}^ \circ }} ,} \\ 
{\sum\limits_{i \in {U_j}} {{{\!\!\!\mathbbm{1}_{\!\left\{ \frac{{{U_j}{H_{ia}}}}{{{R_s}}} + \frac{{G_{{120}^ \circ}}}{{c_{ji}^{SD}}} + \frac{A}{{c_{ji}^{SU}}} \!\right\}}} \!\!-\!N_\textrm{FU} } ,{f_{j{a_{it}}}}\! =\! {{360}^ \circ },z_{jia_{it}}\!=\!120^\circ},} \\ 
{\sum\limits_{i \in {U_j}} { {{\!\!\!\mathbbm{1}_{\!\left\{ \!\frac{{{U_j}{H_{ia}}}}{{{R_s}}} + \frac{{G_{{120}^ \circ}}}{{c_{ji}^{SD}}} + \frac{A}{{c_{ji}^{SU}}} \!\right\}}} \!\!-\!N_\textrm{FS}  },{f_{j{a_{it}}}}\! = \!{{360}^ \circ }, z_{jia_{it}}\!=\!360^\circ,} }
\end{array}} \right.$ 
where
$N_\textrm{FU}={\mathbbm{1}_{\left\{ \frac{{{H_{ia}}}}{{{R_{\textrm{U}} \mathord{\left/
 {\vphantom {R {{U_j}}}} \right.
 \kern-\nulldelimiterspace} {{U_j}}}}}+\frac{{G_{{120}^ \circ}}}{{c_{kj}^\textrm{VD}}} + \frac{A}{{c_{kj}^\textrm{VU}}} + \frac{{G_{{120}^ \circ}}}{{c_{ji}^\textrm{SD}}} + \frac{A}{{c_{ji}^\textrm{SU}}} \leqslant D \right\}}}$ and $N_\textrm{FS}={\mathbbm{1}_{\left\{ \frac{{{H_{ia}}}}{{{R_{\textrm{S}} \mathord{\left/
 {\vphantom {R {{U_j}}}} \right.
 \kern-\nulldelimiterspace} {{U_j}}}}}+\frac{{G_{{120}^ \circ}}}{{c_{kj}^\textrm{VD}}} + \frac{A}{{c_{kj}^\textrm{VU}}} + \frac{{G_{{120}^ \circ}}}{{c_{ji}^\textrm{SD}}} + \frac{A}{{c_{ji}^\textrm{SU}}} \leqslant D \right\}}}$}.
\end{theorem} 
  
  \begin{proof} See Appendix B.
 \end{proof}

From Theorem \ref{theorem2}, we can see that the cached contents of each SBS $j$ depend on the data rates of backhaul and users, the number of users associated with SBS $j$, the computational resources allocated to each user, and the cache storage size. Theorem \ref{theorem2} shows that, if the cache storage size is unlimited, storing each visible content for each user is better than storing a $360^\circ$ content. This is due to the fact that an SBS that stores visible contents can directly send its cached visible contents to its users without performing any backhaul transmission or visible content extraction. When the cache storage is limited, the SBSs may store $360^\circ$ contents to save cache storage. This, in turn, will increase the processing delay. In consequence, the SBSs need to balance the tradeoff between saving cache storage and increasing processing delay. Based on Theorem \ref{theorem2}, each SBS $j$ calculates the number of users that have successful transmission so as to calculate the users' reliability.

\subsection{ELSM for Reliability Maximization}
Next, we explain how each SBS $j$ uses the ELSM algorithm to solve problem (\ref{eq:max}). We first introduce the policy that each SBS $j$ uses for  action selection at each time slot $t$. Then, the use of ELSM algorithm for users' reliability estimation is specified.
\subsubsection{Action Selection} At each time slot $t$, the ELSM algorithm will implement one action chosen based on the Boltzmann policy, which is given by $\boldsymbol{p}_j\!=\!\left[{p}_j\left(\boldsymbol{o}_{j1} \right),\ldots,{p}_j\left(\boldsymbol{o}_{jN_{j}^\textrm{O}} \right)  \right]$ with
\begin{equation}\label{eq:p}
p_j\left( {{\boldsymbol{o}_{ji}}} \right) = \frac{{{e^{{{y\left(\boldsymbol{o}_{ji}\right)} \mathord{\left/
 {\vphantom {{{y_j}\left( {{o_{ji}}} \right)} \kappa }} \right.
 \kern-\nulldelimiterspace} \kappa }}}}}{{\sum\limits_{\boldsymbol{\theta}  \in {\mathcal{A}_j},\boldsymbol{\theta}  \ne {\boldsymbol{o}_{ji}},} {{e^{{{y\left(\boldsymbol{\theta}\right)} \mathord{\left/
 {\vphantom {{{y_j}\left( \theta  \right)} \kappa }} \right.
 \kern-\nulldelimiterspace} \kappa }}}} }},
\end{equation}
where $\mathcal{A}_j$ is the set of actions of SBS $j$ and $\kappa$ is a constant that determines the probability distribution of the Boltzmann policy. For example, as $\kappa  \to \infty $, the Boltzmann policy will follow a uniform distribution.

\subsubsection{Learning Process} At each period $\tau$, the SBSs will update their policies and broadcast them to other SBSs. Then, the SBSs will use the policies of other SBSs as ELSM input. Using ESN, the ELSM algorithm can record the historical information related to the input of the ELSM. The information is recorded in the states of reservoir neurons which is given by:   
 \begin{equation}\label{eq:state}
{\boldsymbol{\mu}_{j}\left(\tau\right)} ={\mathop{f}\nolimits}\!\left( {\boldsymbol{W}_j{\boldsymbol{\mu}_{j}\left(\tau-1\right)} + \boldsymbol{W}_j^\textrm{in}\boldsymbol{\varphi }_j\left(\tau\right)} \right).
\end{equation}

From (\ref{eq:state}), we can see that, the reservoir neurons' states depend on the current input and historical states. Hence, the reservoir neurons' states can record historical information. Based on these reservoir states, the ELSM algorithm can predict the reliability of the users:  
\begin{equation}\label{eq:update}
\boldsymbol{y}_{j}\left(t\right) = {\boldsymbol{W}_{j}^\textrm{out}\left(t\right)}\left[ {\begin{array}{*{20}{c}}
  \boldsymbol{\mu}_{j}\left(\tau\right) \\ 
  \boldsymbol{\varphi }_j\left(\tau\right)
\end{array}} \right],
\end{equation}
where $\boldsymbol{y}_{j}\left(t\right) =\left[ {{y} _{j}\left(\boldsymbol{o}_{j1}\left(t\right)\right)}, \ldots, {{y} _{j}\left(\boldsymbol{o}_{jN_j^O}\left(t\right)\right)}\right]$ is the predicted output of the ELSM algorithm. To predict the users' reliability, the ELSM algorithm must train its output weight matrix $ {\boldsymbol{W}_{j}^\textrm{out}\left(t\right)}$, as follows:
\begin{equation}\label{eq:w}
{\boldsymbol{W}_{ji}^\textrm{out}\left(t\right)} = {\boldsymbol{W}_{ji}^\textrm{out}\left(t\right)} + {\lambda}^\alpha \left( {\sum\limits_{i \in \mathcal{U}_j}  { {\mathbbm{P}}_i}\!\left( \boldsymbol{o}_{ji}\left(t\right) \right)- {{y} _{j}\left(\boldsymbol{o}_{ji}\left(t\right)\right)}} \right){\boldsymbol{\mu}_{j}^{\mathrm{T}}\left(\tau\right)},
\end{equation}
where $\lambda^\alpha$ is the learning rate, ${ {\mathbbm{P}}_i}\!\left( \boldsymbol{o}_{ji}\left(t\right) \right)$ is the reliability calculated based on Theorems \ref{theorem1} and \ref{theorem2}. 
Table II summarizes the proposed ELSM algorithm.

\begin{table}[!t]\label{tb1}
  \centering
  \caption{
    \vspace*{-0.3em} ELSM-based Deep RL Algorithm for resource Allocation}\vspace*{-1em}
    \begin{tabular}{p{3.5in}}
      \hline \vspace*{-0.8em}
     \hspace*{0em}\begin{itemize}\vspace*{-0.1em}
\item[] \hspace*{0em} \textbf{for} each time $\tau$ \textbf{do}.
\item[] \hspace*{1em}(a) Estimate the value of the reliability of the users based on (\ref{eq:update}).
\item[] \hspace*{0.5em} \textbf{if} $\tau=1$ 
\item[] \hspace*{2em}(b) Set the policy of the action selection $\boldsymbol{p}_{j}\left(1\right)$ uniformly.
\item[] \hspace*{0.5em} \textbf{else}
\item[] \hspace*{2em}(c) Set the policy of the action selection $\boldsymbol{p}_{j}\left(\tau\right)$ based on (\ref{eq:p}).
\item[] \hspace*{0.5em} \textbf{end if}
\item[] \hspace*{1em}(d) Broadcast the index of the action selection policy to other SBSs.
\item[] \hspace*{1em}(e) Receive the index of the action selection policy as input $\boldsymbol{\varphi}_{\tau,j}$.
\item[] \hspace*{1em}(f) Sample the states of the spiking neurons in the liquid model as
 \item[] \hspace*{2.1em}   the input of ESN.
\item[] \hspace*{0.5em} \textbf{for} each time $t$ \textbf{do}.
\item[] \hspace*{1.5em}(g) Perform an action based on $\boldsymbol{p}_{j\left(\tau\right)}$ and calculate the actual
\item[] \hspace*{2.6em} reliability of the users. 
\item[] \hspace*{1.5em}(h) Update the states of the reservoir neurons based on (\ref{eq:state}).
\item[] \hspace*{1.5em}(i) Update the output weight matrix based on (\ref{eq:w}).
\item[] \hspace*{0.5em} \textbf{end for}
\item[] \hspace*{0em} \textbf{end for}
\end{itemize}\vspace*{-0.3cm}\\
   \hline
    \end{tabular}\label{tab:algo}\vspace{-0.4cm}
\end{table}

\subsection{Convergence of the ELSM Algorithm}
For the proposed algorithm, we only need to train the output weight matrix of ESN. In consequence, we can directly use the result of \cite[Theorem 2]{Chen2016Echo} which has proved that, for each action $\boldsymbol{o}_{ji}$, ESNs will converge to $\sum\limits_{i \in \mathcal{U}_j}  {{\mathbbm{P}}_i}\!\left( \boldsymbol{o}_{ji} \left(t\right) \right)$ where $\mathcal{A}_{-j} = \prod\nolimits_{k \ne j, k \in \mathcal{B} } {{\mathcal{A}_{k}}}$ is the set of actions other than the action of SBS $j$. To improve the value of $\sum\limits_{i \in \mathcal{U}_j}  {{\mathbbm{P}}_i}\!\left( \boldsymbol{o}_{ji} \left(t\right) \right)$, the ELSM algorithm uses Boltzmann policy for action selection instead of the work in \cite{Chen2016Echo} that uses a greedy mechanism. We assume that  the users' reliability resulting from the ELSM algorithm that uses a Boltzmann policy or a greedy mechanism for action selection are $\sum\limits_{i \in \mathcal{U}_j}  {{\mathbbm{P}}_i}\!\left( \boldsymbol{o}_{j}  \right)$ and $\sum\limits_{i \in \mathcal{U}_j}  {{\mathbbm{P}}_i^\varepsilon}\!\left( \boldsymbol{o}_{j}  \right)$, respectively.
Then the relationship between $\sum\limits_{i \in \mathcal{U}_j}  {{\mathbbm{P}}_i}\!\left( \boldsymbol{o}_{j}  \right)$ and $\sum\limits_{i \in \mathcal{U}_j}  {{\mathbbm{P}}_i^\varepsilon}\!\left( \boldsymbol{o}_{j}  \right)$ can be given by the following theorem. 

  \begin{theorem}\label{theorem3}
\emph{Given the Boltzmann policy in (\ref{eq:p}) and the greedy mechanism in \cite{Chen2016Echo}, the relationship between $\sum\limits_{i \in \mathcal{U}_j}  {{\mathbbm{P}}_i}\!\left( \boldsymbol{o}_{j}  \right)$ and $\sum\limits_{i \in \mathcal{U}_j}  {{\mathbbm{P}}_i^\varepsilon}\!\left( \boldsymbol{o}_{j}  \right)$ is given by:}

\begin{itemize}
\item[\romannumeral1)]\emph{If $\kappa\! \to\! \infty $ and $\varepsilon\!=\!1$, then $\sum\limits_{i \in \mathcal{U}_j}  {{\!\mathbbm{P}}_i^\varepsilon}\!\left( \boldsymbol{o}_{j}  \right)\!=\! \sum\limits_{i \in \mathcal{U}_j}  {{\!\mathbbm{P}}_i}\!\left( \boldsymbol{o}_{j} \right)\!, \boldsymbol{o}_{j} \in \mathcal{A}_{j}$.}
\item[\romannumeral2)]\emph{If $\kappa \!=\! 0$ and $\varepsilon\!=\!0$, then $\sum\limits_{i \in \mathcal{U}_j}  {{\!\mathbbm{P}}_i^\varepsilon}\!\left( \boldsymbol{o}_{j}  \right)= \sum\limits_{i \in \mathcal{U}_j}  {{\!\mathbbm{P}}_i}\!\left( \boldsymbol{o}_{j} \right),\boldsymbol{o}_{j} \in \mathcal{A}_{j}$.}
\item[\romannumeral3)]\emph{If $p_j\left(\boldsymbol{o}_{j}^*\right) > p_j^\varepsilon \left(\boldsymbol{o}_{j}^*\right) $, then $\sum\limits_{i \in \mathcal{U}_j}  {{\mathbbm{P}}_i^\varepsilon}\!\left( \boldsymbol{o}_{j}^* \right)< \sum\limits_{i \in \mathcal{U}_j}  {{\mathbbm{P}}_i}\!\left( \boldsymbol{o}_{j}^* \right)$.}
\end{itemize}
\emph{where $\boldsymbol{o}_{j}^*$ is the optimal action of each SBS $j$, $p_j\left(\boldsymbol{o}_{j}^*\right)$ is the probability that action $\boldsymbol{o}_{j}^*$ will be implemented in the Boltzmann policy, and $p_j^\varepsilon \left(\boldsymbol{o}_{j}^*\right)$ is the probability that action $\boldsymbol{o}_{j}^*$ will be implemented in the greedy policy.}

%

\end{theorem} 
  
  \begin{proof} See Appendix C.
  \end{proof}
From Theorem \ref{theorem3}, we can see that the value at the convergence point stemming from the proposed algorithm depends on the action selection policy used by the learning algorithm. As $\kappa$ in (\ref{eq:p}) increases, the proposed algorithm will use more iterations for exploration and the number of iterations for exploitation decreases. Hence, the system performance will decrease and the convergence speed increases. As $\kappa$ decreases, the proposed learning algorithm will use more iterations for exploitation and less iterations for exploration, and, then, system performance (i.e., reliability) will increase while the convergence speed decreases. In consequence, we need to appropriately adjust $\kappa$ to balance the exploration and exploitation so as to balance the convergence speed and the reliability that the proposed algorithm can achieve.  
 
\section{Simulation Results}
For our simulations, we consider a circular network area having a radius $r=500$ m, $U = 20$ uniformly distributed users, $V = 5$ uniformly distributed UAVs, and $B = 5$ uniformly distributed SBSs. The other parameters used in simulations are listed in Table~II. For comparison purposes, we use two baselines: Q-learning in \cite{30} and ESN based learning algorithm in \cite{Chen2016Echo}.

\begin{table}\footnotesize
  \newcommand{\tabincell}[2]{\begin{tabular}{@{}#1@{}}#2\end{tabular}}
\renewcommand\arraystretch{1}
 \caption{
    \vspace*{-0.05em}SYSTEM PARAMETERS}\vspace*{-0.6em}
\centering  
\begin{tabular}{|c|c|c|c|c|c|}
\hline
\textbf{Parameter} & \textbf{Value} & \textbf{Parameter} & \textbf{Value}&\textbf{Parameter} & \textbf{Value} \\
\hline
$\mu_\textrm{LoS}$&2  &$P_U$& 20 dBm&$G_{120^\circ}$ & 12.5 Mbits \\
\hline
 $\beta$ & 2 & $R_\textrm{S}$& 2 Gbits&  $B_\textrm{VD}$&2 GHz \\
\hline
 $\mu_\textrm{NLoS}$& 2.4 & $X$ &  11.9 &$P _{B}$ & 30 dBm\\
\hline
$ \eta $ &100& $d_0$ & 5 m&$\sigma^2$ & -105 dBm\\
\hline
$N_W$&100& $R_U$ & 1 Gbit&$A$ & 50 kbits \\
\hline
$N_\tau$&10  & $D$ & 20 ms&$G_{360^\circ}$ & 50 Mbits\\
\hline
$\chi _{\sigma_\textrm{LoS}}$ & 5.3&$P_V$ & 20 dBm & $B_\textrm{SD}$ & 16 MHz \\
\hline
 $C_k$& 3& $B_\textrm{SU}$ & 4 MHz  &$f_c$& 38 GHz \\ 
\hline
 $N_L$ & 125 & $\chi _{\sigma_\textrm{NLoS}}$ & 5.27&$B_\textrm{VU}$&500 MHz \\ 
\hline
$Y$&0.13 &$ \rho $ & 30 ms&  $S$&300 Mbits \\
\hline
%
\end{tabular}
 \vspace{-0.3cm}
\end{table}

\begin{figure}
\centering
\vspace{0cm}
\subfigure[]{
\label{figure11} 
\includegraphics[width=7.5cm]{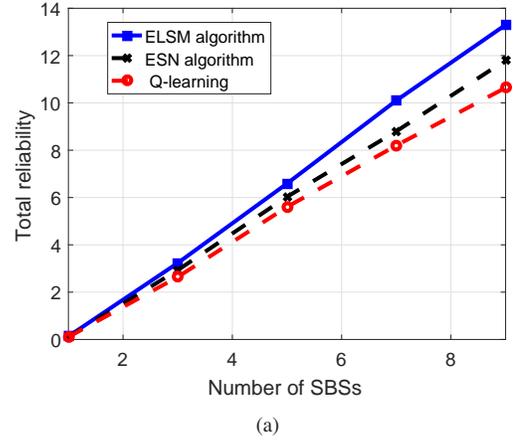}}
\subfigure[]{ 
\label{figure12} 
\includegraphics[width=7.5cm]{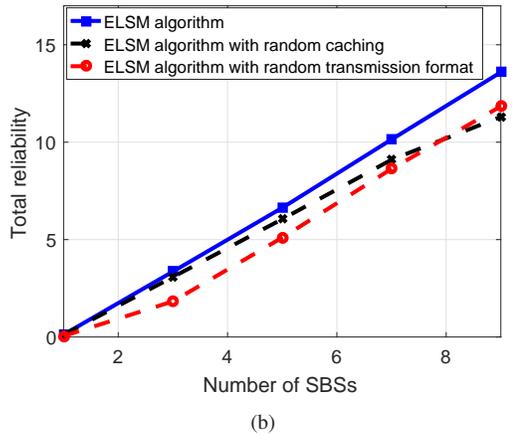}}
  \vspace{-0.25cm}
 \caption{\label{figure1}The total reliability as the number of SBSs varies. }
 \vspace{-0.7cm}
\end{figure}

Fig. \ref{figure1} shows how the total (sum) reliability of all users changes as the number of SBSs varies. From Fig. \ref{figure1}, we can see that the reliability of users increases as the number of SBSs increases. This is due to the fact that, as the number of SBSs increases, the users will have more choices to associate with SBSs and the number of contents stored in the cache per SBS increases. Fig. \ref{figure11} shows that the proposed algorithm yields 14.8\% and 25.4\% gains in terms of total reliability compared to the ESN algorithm and Q-learning. This is because the ELSM algorithm uses a liquid model to record more historical information for the prediction of the total reliability compared to the ESN algorithm and Q-learning. In consequence, the ELSM algorithm can accurately predict the reliability and find a better solution. Fig. \ref{figure12} shows that the proposed algorithm can achieve 14.7\% and 20.2\% gains in terms of reliability compared to the ELSM algorithm with random caching and ESLM with random transmission format. These gains stem from the fact that the proposed algorithm can use Theorems \ref{theorem1} and \ref{theorem2} to find the optimal contents to cache for each SBS as well as the transmission format for each UAV. Fig. \ref{figure12} also shows that, as the number of SBSs increases, the reliability resulting from the ELSM algorithm with random transmission format is larger than the reliability of ELSM algorithm with random caching. This is due to the fact that, as the number of SBSs increases, the number of contents stored in the cache per SBS increases and, hence, the number of contents transmitted from the UAVs to SBSs decreases.

\begin{figure}[!t]
  \begin{center}
   \vspace{0cm}
    \includegraphics[width=7.5cm]{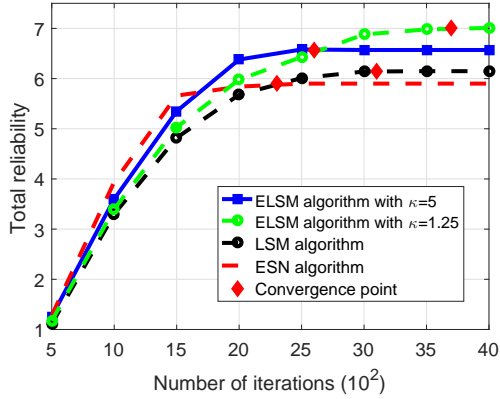}
    \vspace{-0.25cm}
    \caption{\label{figure3} Convergence of learning algorithms.}
  \end{center}\vspace{-0.6cm}
\end{figure}

Fig. \ref{figure3} shows the number of iterations needed till convergence for the proposed ELSM algorithms with $\kappa=5$ and $\kappa=1.25$, the LSM algorithm in \cite{maass2010liquid}, and the ESN algorithm. From this figure, we can see that, as time elapses, the reliability of users increases until convergence to their final values. Fig. \ref{figure3} also shows that the ELSM algorithm with $\kappa=5$ converges faster than the ELSM algorithm with $\kappa=1.25$. Meanwhile, the reliability resulting from the ELSM algorithm with $\kappa=5$ is lower than the reliability of the ELSM algorithm with $\kappa=1.25$. This is due to the fact that, as $\kappa$ increases, the ELSM algorithm uses more iterations for exploration and less iterations for exploitation. In Fig. \ref{figure3}, we can also see that the ELSM approach with $\kappa=5$ needs 2600 iterations to reach convergence and exhibits a considerable reduction of 16.1\% less iterations compared to LSM algorithm. This is because the proposed algorithm uses an ESN as the output function and an ESN only needs to train the output weight matrix. Fig. \ref{figure3} also shows that the proposed algorithm needs 13\% more iterations to reach convergence compared to the ESN algorithm. Meanwhile, the proposed algorithm can achieve 12\% gain in terms of reliability compared to the ESN algorithm. This stems from the fact that the ELSM algorithm uses a liquid model to record more historical information compared to ESN algorithm. Hence, the ELSM algorithm must use more iterations to record network information so as to accurately predict the reliability of users.

\begin{figure}[!t]
  \begin{center}
   \vspace{0cm}
    \includegraphics[width=7.5cm]{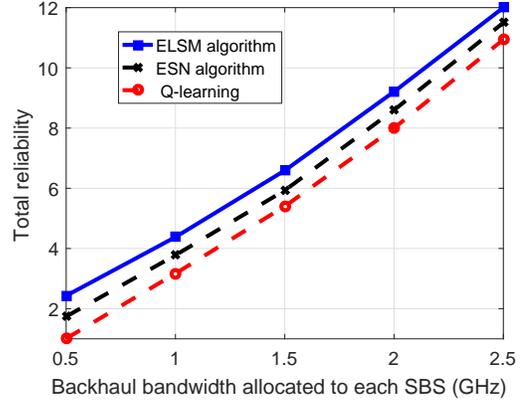}
    \vspace{-0.25cm}
    \caption{\label{figure4} The total reliability as the backhaul bandwidth per SBS varies.}
  \end{center}\vspace{-0.7cm}
\end{figure}
In Fig. \ref{figure4}, we show how the total reliability of users changes as the backhaul bandwidth per SBS varies. From Fig. \ref{figure4}, we can see that, as the backhaul bandwidth increases, the reliability resulting from all of the considered algorithms increases. This is due to the fact that, as the backhaul bandwidth increases, both the delay of transmitting contents from the UAVs to SBSs and the delay of transmitting tracking information from SBSs to the UAVs decrease. 
 Fig. \ref{figure4} also shows that, as the backhaul bandwidth increases, the gains of reliability achieved by the proposed ELSM algorithm compared to the ESN algorithm and Q-learning decrease. This is due to the fact that the UAVs have enough bandwidth for content and tracking information transmission. In consequence, the impact of caching and content transmission format on reliability decreases.         
\begin{figure}[!t]
  \begin{center}
   \vspace{0cm}
    \includegraphics[width=7.5cm]{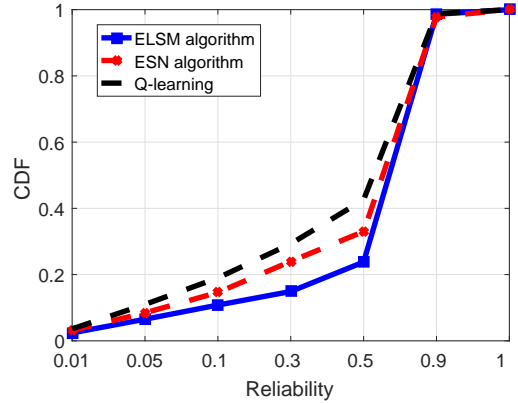}
    \vspace{-0.25cm}
    \caption{\label{CDF} CDFs of the per user reliability resulting from different algorithms.}
  \end{center}\vspace{-0.7cm}
\end{figure}

Fig. \ref{CDF} shows the cumulative distribution function (CDF) for the per user reliability resulting from all of the considered algorithms. From Fig. \ref{CDF}, we can see that the ELSM algorithm improves the CDF of up to 30\% and 48.8\% gains at a reliability of 0.5. This is due to the fact that the ELSM algorithm can record more historical information related to the states of network and users compared to the ESN algorithm and Q-learning. Hence, the proposed algorithm can predict the reliability more accurately compared to the ESN algorithm and Q-learning and find a better solution for the reliability maximization.   
\begin{figure}[!t]
  \begin{center}
   \vspace{0cm}
    \includegraphics[width=7.5cm]{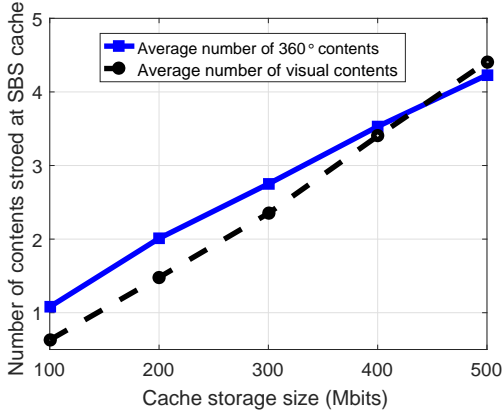}
    \vspace{-0.25cm}
    \caption{\label{figure6} Number of $360^\circ$ and visible contents stored in the cache per SBS as the cache storage size varies.}
  \end{center}\vspace{-0.7cm}
\end{figure}

In Fig. \ref{figure6}, we show how the number of $360^\circ$ and visible contents stored in the cache per SBS changes as the cache storage size varies.  Fig. \ref{figure6} shows that, as the cache storage size per SBS increases, both the number of cached $360^\circ$ and visible contents increases. Fig. \ref{figure6} also shows that, as the cache storage size increases, the number of cached visible contents increases faster than the number of cached $360^\circ$ contents. This is because as each SBS has enough cache storage size, storing visible contents can reduce both the processing delay and the delay of transmitting contents from UAVs to SBSs. In contrast, storing $360^\circ$ contents can reduce only the delay of content transmission over UAV-SBS links.       

\section{Conclusion}
In this paper, we have developed a novel framework that uses flying UAVs to collect VR content for wireless transmission. In this model, the UAVs can transmit VR contents to ground SBSs, over wireless backhaul links. Meanwhile, the SBSs can decide on whether to request and store visible or 360 content from the UAVs, so as to reduce backhaul traffic. We have formulated an optimization problem that jointly considers content caching, content transmission format, and processing delay for users' reliability maximization. To solve this problem, we have developed a novel deep learning algorithm that brings together LSM and ESNs. The proposed algorithm enables each SBS to predict the users' reliability so as to find the optimal contents to cache and content transmission format for each UAV. 
Simulation results have shown that the proposed approach yields significant performance gains. Moreover, the results have also shown that the use of ESN for LSM can significantly improve convergence time when compared to LSM algorithm.

\section*{Appendix}
\subsection{Proof of Theorem \ref{theorem1}} \label{Ap:a} 
To calculate the number of users that have successful transmissions, we need to simplify the problem in (\ref{eq:max}). Given a vector $\boldsymbol{o}_{ji}\left(t\right)$, (\ref{eq:max}) can be simplified as:
  \begin{equation}\nonumber
  \begin{split}
& \!\!\!\!\!\mathop {\max }\limits_{\boldsymbol{g}_i,\boldsymbol{f}_j,{\mathcal{U}_j},\mathcal{S}_j}  \sum\limits_{i \in \mathcal{U}_j}  {{\mathbbm{1}_{\left\{ {T_{it}^\textrm{P}\left( a_{it}, {g}_{ia_{it}}, \boldsymbol{f}_j, U_j, \mathcal{S}_j \right) + T_{it}^\textrm{M}\left( {a_{it},{{g}_{ia_{it}}},{\mathcal{S}_j},{\mathcal{U}_j}} \right) \le D} \right\}}}}\\ 
 \mathop  = \limits^{\left( a \right)} &\mathop {\max }\limits_{\boldsymbol{g}_i  }  \sum\limits_{i \in \mathcal{U}_j} {{\mathbbm{1}_{\left\{ {T_{it}^\textrm{P}\left( a_{it}, {g}_{ia_{it}}, \boldsymbol{o}_{ji}\left(t\right)\right) + T_{it}^\textrm{M}\left( {a_{it},{{g}_{ia_{it}}},\boldsymbol{o}_{ji}\left(t\right)} \right) \le D} \right\}}}},\\
 \mathop  = \limits^{\left( b \right)} & \sum\limits_{i \in \mathcal{U}_j, a_{it}\in \mathcal{S}_j }\!\!\!\! {{\mathbbm{1}_{\left\{ {T_{it}^\textrm{SP}\left(a_{it}, \boldsymbol{o}_{ji}\left(t\right)\right) + \frac{{G_{{120}^ \circ}}}{{c_{ji}^\textrm{SD}}} + \frac{A}{{c_{ji}^\textrm{SU}}} \le D} \right\}}}} \\&+ \mathop {\max }\limits_{\boldsymbol{g}_i  }\!\!\sum\limits_{i \in \mathcal{U}_j, a_{it}\notin \mathcal{S}_j }\!\!\!\! {{\mathbbm{1}_{\left\{ {T_{it}^\textrm{P}\left( a_{it}, {g}_{ia_{it}}, \boldsymbol{o}_{ji}\left(t\right)\right) + T_{it}^\textrm{M}\left( {a_{it},{{g}_{ia_{it}}},\boldsymbol{o}_{ji}\left(t\right)} \right) \le D} \right\}}}},
   \end{split}
 \end{equation}
where (a) is obtained from the fact that action $ \boldsymbol{o}_{ji}\left(t\right)$ determines the user association, cached contents, and cached content format. In consequence, for a given action, we only need to find the optimal transmission format ${\boldsymbol{g}_i}$ for each user $i$. (b) is obtained from the fact that, when the contents that users request are stored at SBS cache, $T_{it}^\textrm{P}\left( a_{it}, {g}_{ia}, \boldsymbol{o}_{ji}\left(t\right)\right)=T_{it}^\textrm{SP}\left(a_{it}, \boldsymbol{o}_{ji}\left(t\right)\right)$ and $ T_{it}^\textrm{M}\left( {a_{it},{{g}_{ia}},\boldsymbol{o}_{ji}\left(t\right)} \right)=\frac{{G_{{120}^ \circ}}}{{c_{ji}^\textrm{SD}}} + \frac{A}{{c_{ji}^\textrm{SU}}}$, which are not related to the content transmission format $\boldsymbol{g}_i$. For each content $a_{it}$ that user $i$ requests at time $t$, if $g_{ia_{it}}=120^\circ$ and $a_{it} \notin \mathcal{S}_j $, the total time used for transmission and processing can be given by:
 \begin{equation}\label{eq:TT1}
 \begin{split}
 &T_{it}^\textrm{P}\left( a_{it}, {g}_{ia_{it}}, \boldsymbol{o}_{ji}\left(t\right)\right) + T_{it}^\textrm{M}\left( {a_{it},{{g}_{ia_{it}}},\boldsymbol{o}_{ji}\left(t\right)} \right)\\&=T_{it}^\textrm{UP}\left(a_{it},g_{ia_{it}}, \boldsymbol{o}_{ji}\right)+T_{it}^\textrm{M}\left( {a_{it},{{g}_{ia_{it}}},\boldsymbol{o}_{ji}\left(t\right)} \right)\\
& =\frac{{{H_{ia}}}}{{{R_{\textrm{U}} \mathord{\left/
 {\vphantom {R {{U_j}}}} \right.\kern-\nulldelimiterspace} {{U_j}}}}}+\frac{{G_{{120}^ \circ}}}{{c_{kj}^\textrm{VD}}} + \frac{A}{{c_{kj}^\textrm{VU}}} + \frac{{G_{{120}^ \circ}}}{{c_{ji}^\textrm{SD}}} + \frac{A}{{c_{ji}^\textrm{SU}}}.
 \end{split}
 \end{equation} 
  Similarly, if $g_{ia_{it}}=360^\circ$ and $a_{it} \notin \mathcal{S}_j $, then the total time used for transmission and processing is given by:
   \begin{equation}\label{eq:TT2}
 \begin{split}
 &T_{it}^\textrm{P}\left( a_{it}, {g}_{ia_{it}}, \boldsymbol{o}_{ji}\left(t\right)\right) + T_{it}^\textrm{M}\left( {a_{it},{{g}_{ia_{it}}},\boldsymbol{o}_{ji}\left(t\right)} \right)\\&=T_{it}^\textrm{SP}\left(a_{it},g_{ia_{it}}, \boldsymbol{o}_{ji}\right)+T_{it}^\textrm{M}\left( {a_{it},{{g}_{ia_{it}}},\boldsymbol{o}_{ji}\left(t\right)} \right)\\
& =\frac{{{H_{ia}}}}{{{R_\textrm{S} \mathord{\left/
 {\vphantom {R {{U_j}}}} \right.
 \kern-\nulldelimiterspace} {{U_j}}}}}+\frac{{G_{{360}^ \circ}}}{{U_{ja}c_{kj}^\textrm{VD}}} + \frac{A}{{c_{kj}^\textrm{VU}}} + \frac{{G_{{120}^ \circ}}}{{c_{ji}^\textrm{SD}}} + \frac{A}{{c_{ji}^\textrm{SU}}}.
 \end{split}
 \end{equation} 
 Since $R_\textrm{U} \leqslant R_\textrm{S}$, then $\frac{{{H_{ia}}}}{{{R_{\textrm{U}} \mathord{\left/
 {\vphantom {R {{U_j}}}} \right.\kern-\nulldelimiterspace} {{BU_j}}}}} >\frac{{{H_{ia}}}}{{{R_\textrm{S} \mathord{\left/
 {\vphantom {R {{U_j}}}} \right.
 \kern-\nulldelimiterspace} {{U_j}}}}}$. Hence, from (\ref{eq:TT1}) and (\ref{eq:TT2}), we can see that, for user $i$, if $\frac{{G_{{360}^ \circ}}}{{U_{ja}c_{kj}^\textrm{VD}}}\leqslant\frac{{G_{{120}^ \circ}}}{{c_{kj}^\textrm{VD}}} \Rightarrow  \frac{{G_{{360}^ \circ}}}{{U_{ja}}} \leqslant {{G_{{120}^ \circ}}}$, then UAV $j$ must transmit $360^\circ$ content $a$ to user $i$, $g_{ia_{it}}=360^\circ$. If $\frac{{G_{{360}^ \circ}}}{{U_{ja}}}>{{G_{{120}^ \circ}}}$ and $\frac{{{H_{ia}}}}{{{R_\textrm{S} \mathord{\left/
 {\vphantom {R {{U_j}}}} \right.
 \kern-\nulldelimiterspace} {{U_j}}}}}+\frac{{G_{{360}^ \circ}}}{{U_{ja}c_{kj}^\textrm{VD}}}-\frac{{{H_{ia}}}}{{{R_{\textrm{U}} \mathord{\left/
 {\vphantom {R {{U_j}}}} \right.\kern-\nulldelimiterspace} {{BU_j}}}}}-\frac{{G_{{120}^ \circ}}}{{c_{kj}^\textrm{VD}}}=H_{ia}U_j \left(\frac{1}{R_\textrm{S}}-\frac{B}{R_\textrm{U}} \right)+\frac{1}{c_{kj}^\textrm{VD}} \left( \frac{{G_{{360}^ \circ}}}{{U_{ja}}}-G_{{120}^ \circ}\right) <0$, then $g_{ia_{it}}\!=\!360^\circ$. As  $\frac{{G_{{360}^ \circ}}}{{U_{ja}}}>{{G_{{120}^ \circ}}}$ and $H_{ia}U_j \left(\frac{1}{R_\textrm{S}}-\frac{B}{R_\textrm{U}} \right)+\frac{1}{c_{kj}^\textrm{VD}} \left( \frac{{G_{{360}^ \circ}}}{{U_{ja}}}-G_{{120}^ \circ}\right)>0$, $g_{ia_{it}}\!\!=\!\!120^\circ$. This completes the proof.

\subsection{Proof of Theorem \ref{theorem2}} \label{Ap:b} 
Given an action $\boldsymbol{o}_{ji}'\left(t\right)$,  (\ref{eq:max}) can be given by: 
    \begin{equation}\nonumber
  \begin{split}
&\!\!\!\!\!\!\!\mathop {\max }\limits_{\boldsymbol{g}_i,\boldsymbol{f}_j,{\mathcal{U}_j},\mathcal{S}_j}  \sum\limits_{i \in \mathcal{U}_j}\!\!  {{\mathbbm{1}_{\left\{ {T_{it}^\textrm{P}\left( a_{it}, {g}_{ia_{it}}, \boldsymbol{f}_j, U_j, \mathcal{S}_j \right) + T_{it}^\textrm{M}\left( {a_{it},{{g}_{ia_{it}}},{\mathcal{S}_j},{\mathcal{U}_j}} \right) \le D} \right\}}}},\\
 \mathop  = \limits^{\left( a \right)}& \mathop {\max }\limits_{\boldsymbol{f}_j, \mathcal{S}_j  }  \sum\limits_{i \in \mathcal{U}_j} {{\mathbbm{1}_{\left\{ {T_{it}^\textrm{P}\left( a_{it}, \boldsymbol{f}_j, \mathcal{S}_j, \boldsymbol{o}_{ji}'\left(t\right)\right) + T_{it}^\textrm{M}\left( {a_{it},\boldsymbol{f}_j, \mathcal{S}_j,\boldsymbol{o}_{ji}'\left(t\right)} \right) \le D} \right\}}}},\\
 \mathop  = \limits^{\left( b \right)} &\mathop {\max }\limits_{\boldsymbol{f}_j, \mathcal{S}_j  }\!\!  \sum\limits_{i \in \mathcal{U}_j, g_{i{a_{it}}}=120^\circ } {{\!\!\!\!\!\mathbbm{1}_{\left\{ \!{T_{it}^\textrm{P}\!\left( a_{it},\boldsymbol{f}_j, \mathcal{S}_j \boldsymbol{o}_{ji}'\!\left(t\right)\!\right) + T_{it}^\textrm{M}\!\left( {a_{it},\boldsymbol{f}_j, \mathcal{S}_j,\boldsymbol{o}_{ji}'\!\left(t\right)}\! \right) \le\! D\!} \right\}}}}\\ 
 +&\mathop {\max }\limits_{\boldsymbol{f}_j, \mathcal{S}_j  } \!\! \sum\limits_{i \in \mathcal{U}_j, g_{i{a_{it}}}=360^\circ } {{\!\!\!\!\!\mathbbm{1}_{\left\{ \!{T_{it}^\textrm{P}\!\left( a_{it},\boldsymbol{f}_j, \mathcal{S}_j \boldsymbol{o}_{ji}'\!\left(t\right)\!\right) + T_{it}^\textrm{M}\!\left( {a_{it},\boldsymbol{f}_j, \mathcal{S}_j,\boldsymbol{o}_{ji}'\!\left(t\right)}\! \right) \le\! D\!} \right\}}}},
   \end{split}
 \end{equation}
 where (a) is obtained from the fact that action $\boldsymbol{o}_{ji}'\left(t\right)$ determines the user association and content transmission format $\boldsymbol{g}_{i}$. (b) is obtained from the fact that the transmission format $g_{ia_{it}}=z_{jia_{it}}$. For user $i$, if $g_{ia_{it}}=z_{jia_{it}}=120^\circ$, the the total time used for transmission and processing of is:
 \begin{equation}\nonumber
 \begin{split}
  &{T_{it}^\textrm{P}\left( a_{it}, \boldsymbol{o}_{ji}'\left(t\right)\right) + T_{it}^\textrm{M}\left( {a_{it},\boldsymbol{o}_{ji}'\left(t\right)} \right)}=\\& \left\{ {\begin{array}{*{20}{c}}
 \!\!\!\!\!\!\!\!\! \!\!\!\!\!\!  \!\!\!\!\!\!\!\!\!\!\!\!{\frac{{G_{{120}^ \circ}}}{{c_{ji}^\textrm{SD}}} + \frac{A}{{c_{ji}^\textrm{SU}}},{a_{it}} \in {S_j},{f_{ja_{it}}} = {{120}^ \circ },} \\ 
   \!\!\!\!\!\!{\frac{{{U_j}{H_{ia}}}}{{{R_s}}} + \frac{{G_{{120}^ \circ}}}{{c_{ji}^\textrm{SD}}} + \frac{A}{{c_{ji}^\textrm{SU}}},{a_{it}} \in {S_j},{f_{ja_{it}}} = {{360}^ \circ },} \\ 
  {\frac{{{H_{ia}}}}{{{R_{\textrm{U}} \mathord{\left/
 {\vphantom {R {{U_j}}}} \right.
 \kern-\nulldelimiterspace} {{U_j}}}}}+\frac{{G_{{120}^ \circ}}}{{c_{kj}^\textrm{VD}}} + \frac{A}{{c_{kj}^\textrm{VU}}} + \frac{{G_{{120}^ \circ}}}{{c_{ji}^\textrm{SD}}} + \frac{A}{{c_{ji}^\textrm{SU}}}, {a_{it}} \notin {S_j}.} 
\end{array}} \right.
\end{split}
  \end{equation}
The number of users that have successful transmissions due to caching visible content $a_{it}$ can be given by:

  \begin{equation}\label{eq:IF1}
  {I_F}\!\left( {{f_{j{a_{it}}}}} \right)\! = \!\left\{ {\begin{array}{*{20}{c}}
 {\sum\limits_{i \in {U_j}} {{{\mathbbm{1}_{\left\{ \frac{{G_{{120}^ \circ}}}{{c_{ji}^\textrm{SD}}} + \frac{A}{{c_{ji}^\textrm{SU}}} \leqslant D \right\}}}-N_\textrm{FU}},{f_{j{a_{it}}}} = {{120}^ \circ }} ,} \\ 
 {\sum\limits_{i \in {U_j}} { {{\mathbbm{1}_{\!\left\{ \frac{{{U_j}{H_{ia}}}}{{{R_s}}} + \frac{{G_{{120}^ \circ}}}{{c_{ji}^\textrm{SD}}} + \frac{A}{{c_{ji}^\textrm{SU}}} \right\}}} - N_\textrm{FU}} ,{f_{j{a_{it}}}}= {{360}^ \circ }.} } 
\end{array}} \right.
  \end{equation}

 As $g_{ia_{it}}=360^\circ$, the number of users that have successful transmission due to caching $360^\circ$ content $a_{it}$ can be given by:
  \begin{equation}\label{eq:IF2}
  {I_F}\!\left( {{f_{j{a_{it}}}}} \right) \!= \!\left\{ {\begin{array}{*{20}{c}}
 {\sum\limits_{i \in {U_j}} { {{\mathbbm{1}_{\left\{ \frac{{G_{{120}^ \circ}}}{{c_{ji}^\textrm{SD}}} + \frac{A}{{c_{ji}^\textrm{SU}}} \leqslant D\! \right\}}}-N_\textrm{FU}},{f_{j{a_{it}}}} = {{120}^ \circ }} ,} \\ 
 {\sum\limits_{i \in {U_j}} {{{\mathbbm{1}_{\left\{\frac{{{U_j}{H_{ia}}}}{{{R_s}}} + \frac{{G_{{120}^ \circ}}}{{c_{ji}^\textrm{SD}}} + \frac{A}{{c_{ji}^\textrm{SU}}} \right\}}}- N_\textrm{FS} },{f_{j{a_{it}}}} \!=\! {{360}^ \circ }.} } 
\end{array}} \right.
  \end{equation}
  Based on (\ref{eq:IF1}) and (\ref{eq:IF2}), we can find the contents stored in the cache of SBS $j$ and their storage formats to maximize $ {I_F}\left( {{f_{j{a_{it}}}}} \right)$, which is given by $\left[ {{\mathcal{S}_j},{\boldsymbol{f}_j}} \right] = \arg \mathop {\max }\limits_{{\mathcal{S}_j},{\boldsymbol{f}_j}} \sum\limits_{{a_{it}} \in {S_j},{g_{i{a_{it}}}} = {z_{ji{a_{it}}}}} {{I_F}\left( {{f_{j{a_{it}}}}} \right)} $ with the constraint $\sum\limits_{a \in {\mathcal{S}_j},} { \chi_a \left( {{f_{ja}}} \right)}  \le S$. This completes the proof.

\subsection{Proof of Theorem \ref{theorem3}} \label{Ap:c} 
  To prove Theorem \ref{theorem3}, we first note that the greedy mechanism can be given by $p_j^\varepsilon \left(\boldsymbol{o}_{j}^*\right)=1-\varepsilon+\frac{\varepsilon }{{N_j^O}}$ and $p_j^\varepsilon \left(\boldsymbol{o}_{j}\right)=1-\varepsilon+\frac{\varepsilon }{{N_j^O}}, \boldsymbol{o}_{j} \in \mathcal{A}_{j}, \boldsymbol{o}_{j}  \ne  \boldsymbol{o}_{j}^*$. For case \romannumeral1), as $\kappa \to \infty $, $p_j\left( {{\boldsymbol{o}_{ji}}} \right) = \frac{1}{N_j^\textrm{O}}$. As $\varepsilon=1$,  $p_j^\varepsilon \left(\boldsymbol{o}_{j}\right)= \frac{1}{N_j^\textrm{O}}$. Since $\small \sum\limits_{i \in \mathcal{U}_j}  {{\mathbbm{P}}_i}\!\left( \boldsymbol{o}_{ji} \right)=\sum\limits_{\boldsymbol{o}_{-j}  \in \mathcal{A}_{-j}} {F_N\left( \boldsymbol{o}_{ji} , \boldsymbol{o}_{-j} \right)}\prod\limits_{k = 1,k \ne j}^B {{p}_k\left( \boldsymbol{o}_{k} \right)}$, $p_j^\varepsilon \left(\boldsymbol{o}_{j}\right)=p_j \left(\boldsymbol{o}_{j}\right)$ will result in $\sum\limits_{i \in \mathcal{U}_j}  {{\mathbbm{P}}_i^\varepsilon}\!\left( \boldsymbol{o}_{j}  \right)= \sum\limits_{i \in \mathcal{U}_j}  {{\mathbbm{P}}_i}\!\left( \boldsymbol{o}_{j} \right)$. For case \romannumeral2), as $\kappa = 0$ and $\varepsilon=0$, $p_j \left(\boldsymbol{o}_{j}^*\right)=p_j^\varepsilon \left(\boldsymbol{o}_{j}^*\right)=1$ and $p_j \left(\boldsymbol{o}_{j}\right)=p_j^\varepsilon \left(\boldsymbol{o}_{j}\right)=0, \boldsymbol{o}_{j} \in \mathcal{A}_{j}, \boldsymbol{o}_{j}  \ne  \boldsymbol{o}_{j}^*$. Therefore, $\sum\limits_{i \in \mathcal{U}_j}  {{\mathbbm{P}}_i^\varepsilon}\!\left( \boldsymbol{o}_{j} \right)= \sum\limits_{i \in \mathcal{U}_j}  {{\mathbbm{P}}_i}\!\left( \boldsymbol{o}_{j}  \right)$. For case \romannumeral3), as $p_j\left(\boldsymbol{o}_{j}^*\right) > p_j^\varepsilon \left(\boldsymbol{o}_{j}^*\right) $, then $\prod\limits_{k = 1,k \ne j}^B {{p}_k\left( \boldsymbol{o}_{k}^* \right)}>\prod\limits_{k = 1,k \ne j}^B {{p}_k^\varepsilon\left( \boldsymbol{o}_{k}^{*} \right)}$ and, hence, $ \sum\limits_{i \in \mathcal{U}_j}  {{\mathbbm{P}}_i}\!\left( \boldsymbol{o}_{j}^* \right)> \sum\limits_{i \in \mathcal{U}_j}  {{\mathbbm{P}}_i^\varepsilon}\!\left( \boldsymbol{o}_{j}^* \right)$.
  Based on (\ref{eq:p}), $p_j \left(\boldsymbol{o}_{j}\right)$ depends on the value of $\sum\limits_{i \in \mathcal{U}_j}  {{\mathbbm{P}}_i}\left( \boldsymbol{o}_{j} \right)$. When the value of $\sum\limits_{i \in \mathcal{U}_j}  {{\mathbbm{P}}_i}\left( \boldsymbol{o}_{j}  \right)$ is larger than the value of $\sum\limits_{i \in \mathcal{U}_j}  {{\mathbbm{P}}_i}\!\left( \boldsymbol{o}''_{j}  \right)$, $p_j \left(\boldsymbol{o}_{j}\right)>p_j \left(\boldsymbol{o}''_{j}\right), \boldsymbol{o}_{j}, \boldsymbol{o}''_{j} \in \mathcal{A}_{j}, \boldsymbol{o}_{j} \ne \boldsymbol{o}''_{j}, \boldsymbol{o}_{j},\boldsymbol{o}''_{j} \ne \boldsymbol{o}_{j}^*$ while in greedy mechanism, $p_j^\varepsilon \left(\boldsymbol{o}_{j}\right)=p_j^\varepsilon \left(\boldsymbol{o}''_{j}\right)$. This completes the proof.

\bibliographystyle{IEEEbib}
\def\baselinestretch{0.98}
\bibliography{references1}
\end{document}